\documentclass[10pt,reqno]{amsart}
     \makeatletter
     \def\section{\@startsection{section}{1}%
     \z@{.7\linespacing\@plus\linespacing}{.5\linespacing}%
     {\bfseries%\normalfont\scshape
     \centering
     }}
     \def\@secnumfont{\bfseries}
     \makeatother
\newtheorem{theorem}{Theorem}[section]
\newtheorem{lemma}[theorem]{Lemma}

\newtheorem{corollary}[theorem]{Corollary}
\theoremstyle{definition}

\newtheorem{example}[theorem]{Example}
\theoremstyle{remark}

\numberwithin{equation}{section}
\usepackage{amsfonts}
\usepackage[active]{srcltx}
\usepackage{amsmath}
\usepackage{bbm}
\usepackage{amssymb}
\usepackage{amsthm}
\usepackage{mathtools}
\usepackage{shuffle}
\usepackage{bm} %for bold math symbol
\usepackage{upgreek}
\usepackage{mathrsfs}
\usepackage{fancyhdr}
\usepackage{amsthm}
\usepackage{cases}
\usepackage{amsmath,amssymb,bm}
\usepackage{enumitem}
\usepackage[]{showlabels}
\usepackage{xcolor}

\usepackage{tikz-cd} 
\usepackage{amsfonts}

\usepackage{enumitem}
\usepackage{wasysym}
\usepackage{mathscinet}
\usepackage{verbatim}

\usepackage{graphicx}
\usepackage[
bookmarks=true,         %generates bookmarks for all entries in the "Table of Contents"
bookmarksnumbered=true, %includes chapter-/header-/section-/subsection-/\cdots   -
colorlinks=true, pdfstartview=FitV, linkcolor=blue, citecolor=blue,
urlcolor=blue]{hyperref}
\usepackage{microtype}
\theoremstyle{plain}
\newtheorem{problem}[theorem]{Problem}

\renewenvironment{proof}[1][Proof]{\textbf{#1.} }{\ \rule{0.5em}{0.5em} \par }
\theoremstyle{remark}

\theoremstyle{definition}

\def\RR{\mathbb{R}}
\def\TT{\mathbb{T}}
\def\EE{\mathbb{E}}
\def\NN{\mathbb{N}}

\def\PP{\mathbb{P}}
\def\JJ{\mathbb{J}}

\def\cF{{\mathcal F}}

\def\la{{\lambda}}

\def\De{{\Delta}}

\def\cF{{\mathcal F}}
\def\Om{{\Omega}}

\def\al{{\alpha}}

\def\De{{\Delta}}

\def\la{{\lambda}}

\def\th{{\theta}}
\def\th{{\theta}}

\def\al{{\alpha}}

\def\111{{\mathbb{I}}}

\begin{document}

\title[Jump models with delay]{Jump Models with delay - option pricing  
and logarithmic Euler-Maruyama scheme}
\thanks{supported by an NSERC discovery fund and a startup fund of University of Alberta.} 

\author[Agrawal]{Nishant Agrawal} 
\address{Department of Mathematical and Statistical Sciences \\
 University of Alberta at Edmonton \\
Edmonton,  Canada, T6G 2G1}
\email{nagrawal@ualberta.ca, yaozhong@ualberta.ca}

\author[Hu]{Yaozhong Hu} 


\begin{abstract}
In this paper, we 
obtain the existence,  uniqueness and  positivity of the solution
to   delayed stochastic differential equations  with jumps. This equation  is 
then applied  to   model the price movement  of the  risky asset 
in a financial market and the Black-Scholes  formula for the 
price of European option is obtained together with the hedging portfolios.
The option price is evaluated analytically at the last delayed period
by using the Fourier transformation technique. But in general there is no analytical expression for the option price. To evaluate the price numerically we then use the Monte-Carlo method. To this end we need to simulate the delayed 
stochastic differential equations  with jumps. We propose a logarithmic Euler-Maruyama 
scheme to approximate the equation and prove that all the approximations remain positive and the rate of convergence of the scheme is proved  to be 
$0.5$.
\end{abstract}

\maketitle

\section{Introduction} 
The risky asset in the classical Black-Scholes market is 
described by the geometric Brownian motion
given by the stochastic differential equation
driven by standard Brownian motion:
\begin{eqnarray}
dS(t)=S(t) \left[ rdt +\sigma dW(t)\right]\,,
\label{e.1.3}
\end{eqnarray}
where $r$ and $\sigma$ are two positive constants and $W(t)$ is
the standard Brownian motion.  Ever since the seminal work of Black, Scholes and Merton 
there have been many research works to extend the Black-Scholes-Merton's theory 
of option pricing from the original Black-Scholes market   to   more sophisticated models. 
%We need models where instantaneous volatility can be taken into account in the form of history of price of the stock. Hobson and Rogers \cite{hr} suggested a new class of non constant volatility models, which can be extended to
%include the aforementioned level-dependent model and share many characteristics with the stochastic volatility model. Kind et al. \cite{kind} obtained a diffusion approximation result for processes satisfying some
%equations with past-dependent coefficients, and they applied this result to a model of option pricing, in which the
%underlying asset price volatility depends on the past evolution to obtain a generalized (asymptotic) Black–Scholes
%formula.

%Hence it is natural to consider volatility being regarded as function of past states of the stock price. Hence we consider the model of stochastic delay differential equations with jumps.
One of these extensions is the delayed stochastic differential equation (SDDE) 
driven by the standard Brownian motion
(e.g. \cite{Hu2}, see also \cite{mao, aswish}). 
%\footnote{add some other references about this
%
%{\color{blue}
%I have added some references there and a few more elsewhere in the introduction section. 
%}}
  In these works the risky asset is described by the following
 stochastic  delay differential equation 
\[
dS(t)=S(t) \left[ f(t, S_t)dt +g(t, S_t) dW(t)\right]\,,
\]
where $S_t=\{S(s) \,, t-b\le s\le t\}$ 
or $S_t=S(t-b)$ for some constant $b>0$. .

On the other hand, there have been some recent discovery 
(see e.g. \cite{kou2002,kou2004, MixedJDP, First_pass_time_of_double_JDP}) that  to better 
fit some risky assets it is more desirable to use the  hyper-exponential 
jump process along with the classical Brownian motion: 
\[
dS(t)=S(t) \left[ rdt +\sigma dW(t)+\beta dZ(t)\right]\,,
\] 
where $Z(t)$ is a hyper-exponential 
jump process (see   the definition in the next section). 

Let $N(dt, dz)$ be the Poisson random measure associated with 
a jump process which includes the hyper-exponential jump process as a special case and let $\tilde N(dt, dz)$ denote its compensated Poisson random measure. Then
the above equation  with $\sigma=0$ is a special case of the following equation   
\begin{eqnarray}
dS(t)=S(t)\Big(r  dt + \beta \displaystyle\int_{[0,T]\times \mathbb{R}_0}z\tilde{N}(dz,ds)\Big)\label{e.1.2}
\end{eqnarray}
 and it has been argued in (eg. \cite{barndorff,eberlein,cont}) that 
the equation \eqref{e.1.2} is a better model for stock prices than \eqref{e.1.3}.
 
In this paper, we propose a new model to describe the risky asset by combining 
the hyper-exponential process with delay. More precisely, we propose
the following stochastic differential equation as a model for 
the price movement of the risky asset: 
\begin{eqnarray}
dS(t)=S(t) \left[ f(t, S(t-b))dt +g(t, S(t-b)) dZ(t)\right]\,, 
\label{e.1.1}
\end{eqnarray} 
where $f$ and $g$ are two given functions, and $Z(t)$ is a L\'evy 
process which include the hyper-exponential jump processes  as a  special case. The above model along with the Brownian motion 
 component can be found  in \cite{zaheer}, where the   coefficient of Brownian motion cannot be allowed to be zero. In this work, we let the coefficient 
 of the Brownian motion to be zero and we use the Girsanov formula 
 for the jump process to address the issue of completeness of the market and hedging portfolio missed in \cite{zaheer}. 
 
With the introduction of this new market model, the first question is that 
whether  the equation has a unique solution or not and if the  unique 
solution exists
whether the solution is positive or not 
(since the price of an asset is always positive). We shall first 
answer these questions in Section \ref{s.2}, where we prove the 
existence, uniqueness and positivity of the solutions to a larger class of equations than \eqref{e.1.1}.  To guarantee that the solution 
is positive,  we need to assume that  the  jump part $g(t, S(t-b)) dZ(t)$ of the 
equation is bounded from below by some constant (see the assumption (A3) in the next section for the  precise meaning).  
The class of the equations our results can be applied 
is larger in the following two aspects: The first one is   that $Z(t)$ can be replaced by a more general L\'evy  process or more general Poisson random measure and the second one is that the equation can be multi-dimensional. 

Following the Black-Scholes-Merton's principle we  then obtain a formula 
for the fair  price for the European option and the corresponding 
replica hedging portfolio is also given. To evaluate this formula 
during the last delay period, we propose a Fourier transformation method.
This method appears more explicit than the partial differential equation method 
in the literature and is more closed to the original Black-Scholes formula in spirit.  This is done in Section \ref{s.4}. 

Due to the involvement of $f(S(t-b))$ and 
$g(S(t-b))$ the above  analytical expression for the fair option price formula 
is only valid in the last delay period.  Then how do we  perform the evaluation by using this option price formula?  We propose to use  Monte-Carlo
method to get the numerical value approximately.  For this  reason 
we need to simulate the equation \eqref{e.1.1} numerically. 
We observe   that there  have  been  a lot of works 
 (eg.  \cite{kumar, istvan, fuke}) on Euler-Maruyama convergence scheme for SDDE models.   There has already been study on
 the  Euler-Maruyama   scheme for SDDE models
 with jumps (e.g. \cite{convg}).  However,   in general the Euler-Maruyama scheme cannot preserve the positivity of the solution. 
Since the solution to the equation 
\eqref{e.1.1} 
is positive (when the initial condition is positive), 
we wish all of our approximations of the solution is also positive. 
To this end and motivated by the similar work in the Brownian motion case
(see e.g. \cite{huyi})  we introduce a logarithmic Euler-Maruyama scheme, a variant of the Euler-Maruyama scheme for \eqref{e.1.1}. 
With this scheme all the approximate solutions are positive and the rate of the convergence of this scheme is also $0.5$. This rate is optimal
even in the Brownian motion case (e.g. \cite{cambanis}). Let us point out that 
   the $0.5$ rate of the usual Euler-Maruyama scheme  for SDDE with jumps 
   studied in   \cite{convg} 
is only obtained in the  $L^2$   sense.  Not only our logarithmic Euler-Maruyama scheme 
preserves  the positivity, its  rate is $0.5$ in $L^p$ for any $p\ge 2$. 
This is done in Section \ref{s.3}.

% We have derived fair price formula for the option pricing in Delayed Black Scholes model with double exponential  jump diffusion process (DEJP) or hyper exponential jump diffusion process (HEJP). 
% Our work is an extension of \cite{Hu2}. We have considered two different models and calculated the European Call option price for both. We have also shown the convergence of Euler Maruyama scheme in p norm considering the fact that stock price stays positive throughout. Our first model is the jump model which has the jump and a drift component and second model also considers diffusion component in addition to jump and drift. Not many models have been discussed in the litrature where there is no diffusion component. By doing that discussion we beleive that we have come up with a fairly new model (involving no diffusion component along with delay factor) which also looks promising going by the numerical results we have got.
 
% 
% This paper has been organised into following sections. Second section gives the option pricing formula for european option when undelying model is the jump drift model. Third section establishes that solution is indeed unique and positive. Fourth section discusses and presents the proof of convergence of Euler approximation for the jump drift model. We have disscsed the jump diffusion model with delay in section five and derived the formula for the European Call option price we then prove the convergence of the Euler Maruyama scheme in section six for the same. 
Finally in Section \ref{s.5}
%\deleted{5} \added{\ref{s.5}}  
we   present  some numerical attempts 
%of the jump drift model 
and compared that with the classical Black-Scholes price 
formula against the market price for some famous  call options in the 
 real financial 
market. 

\section{Delayed stochastic differential equations with jumps}
\label{s.2} 
%In this section we have derived the price of the european call option where the underlying model follows the jump model with delay.  We have considered finite activity jump in the form of compound poisson process. We have taken a special case of Assymetric double exponential jump process and hyper exponential jump process. The same formula will also work for any process where we can break down our leavy measure $\nu$ as $\nu(dx) = \lambda f(x)ds $ where $\lambda $ is the intensity of the compound poisson process and $f$ is the probability density function of the i.i.d jumps.   
%
% We consider the following model for risk free asset
% $$B(t) = e^{rt}$$
% and for the risky asset.
Let $(\Om, \cF, \PP)$ be a probability space with a filtration $(\cF_t)_{\{t\ge 0\}}$ satisfying the usual conditions. On $(\Om, \cF, \PP)$  let  $Z(t)$ 
be a L\'evy process 
adapted to the filtration $\cF_t$. We shall consider the following delayed  stochastic differential equation  driven by the L\'evy process $Z(t)$:
\begin{eqnarray}
\begin{cases} 
dS(t) = f(S(t-b)) S(t )dt + g(S(t-b))S(t-)d {Z}(t),\hspace{1cm} t \ge 0 \,,  \\
 S(t) = \phi(t)\,,  \hspace{1cm} t \in [-b,0]\,,  
 \end{cases} \label{e.2.1} 
 \end{eqnarray} 
where 
\begin{enumerate}
\item[(i)] 
 $f, g:\RR\rightarrow
\RR$ are some    given bounded measurable  functions;
\item[(ii)]  $b>0$  is  a given number representing the delay of the equation;
\item[(iii)] $\phi:    [-b,0]\rightarrow \RR$ is a 
(deterministic) measurable function. 
\end{enumerate} 
 
To study the above stochastic differential equation,  it is common to introduce the 
Poisson random measure associated with this   L\'evy process $Z(t)$
(see e.g. \cite{d.applebaum, cont, NOP, protter} and references therein). 
First,  we write
 the jump of the process $Z$ at time $t$  by
   \[
    \Delta Z(t):= Z(t) - Z(t-) \quad \hbox{if $ \Delta  Z(t)\not=0$}\,.
    \]
 Denote  $\mathbb{R}_0 := \mathbb{R} \backslash \{0\}$ and let $\mathcal{B}(\mathbb{R}_0)$ be the Borel $\sigma$-algebra generated by the family of
 all Borel subsets $U \subset \mathbb{R}$, such that $\Bar{U} \subset \mathbb{R}_0$. For any 
 $t>0$ and for any  $U \in \mathcal{B}(\mathbb{R}_0)$ 
 %with $\Bar{U} \subset \mathbb{R}_0$ 
  we   define the {\it Poisson random measure},  
 $N: [0, T]\times \mathcal{B}(\mathbb{R}_0)\times \Om\rightarrow \RR$,  associated with 
 the L\'evy process $Z$ by 
 \begin{equation}
 N(t, U) := \sum_{0 \leq s \leq t, \ \Delta Z_s\not =0}\chi_U(\Delta Z(s))\,,    
 \end{equation} 
 where $\chi_U$ is the indicator function of $U$.  
 The associated  L\'evy measure $\nu$ of the L\'evy process $Z$ is  given  by
 \begin{equation}
 \nu(U) := \mathbb{E}[N(1,U)] 
 \end{equation}  
 and the compensated Poisson random  measure $\tilde{N}$ associated with 
 the L\'evy process $Z(t)$  is defined by 
 \begin{equation}
 \Tilde{N}(dt,dz) := N(dt,dz) - \EE\left[ N(dt,dz)  \right] = N(dt,dz) - \nu(dz)dt\,. 
 \end{equation}
For some technical reason, we  shall assume that the process $Z(t)$ has only bounded negative jumps to guarantee that the solution $S(t)$ to \eqref{e.2.1} is positive.
This means that there is an interval $\JJ=[-R, \infty)$ bounded from the left 
such that $\Delta Z(t)
\in \JJ$ for all $t>0$. 
With these notations,  we can write
\[
 {Z}(t) = \displaystyle\int_{ [0,t] \times \JJ } z  {N} (ds,dz) \quad {\rm or}\quad d {Z}(t) = \displaystyle\int_{ \JJ } z  {N} (dt,dz) 
 \]
and the equation  \eqref{e.2.1}  becomes
\begin{eqnarray*}
dS(t)&=& \left[ f(S(t-b)) + g(  S(t-b)) \int_{\JJ} z\nu(dz)\right]
S(t) dt\nonumber\\
&&\qquad\quad  +g(  S(t-b)) S(t-) \int_{\JJ}    z \tilde N(dt, dz) \,. 
\end{eqnarray*}  
It is a   special case 
of the following equation:
\begin{eqnarray}
dS(t)=   f(S(t-b))  
S(t) dt +\int_{\JJ} g(z, S(t-b)) S(t-)  \tilde N(dt, dz) \,.
\label{e.2.5}
\end{eqnarray}
 % = \sum_{i=1}^{P_t}Y_i - \mathbb{E} [\sum_{i=1}^{P_t}Y_i].$$
% and  $\tilde{N} (dz,dt)$ 
% is the compensated poisson random measure with respect to measure $\mathbb{P}$ and $P_t$ is the Poisson process with intensity $\lambda$.  

%      
% \begin{enumerate}
%     \item[I1] $f(t)>0 $  for all $t \in [0,\infty) $.
% \item[I2]  $g :\mathbb{R} \rightarrow \mathbb{R}-0$ is continuous and $|g| \leq B$ where $0<B< \infty $.
% \item[I3] $a$ and $b$ are positive constants.
% \item[I4]{\color{red} $zg(S(t-b))>-1$ for all $t\geq 0$ and $\Big|z\Big|< \frac{1}{B}$.} \item[I5]
% {\color{red}$\displaystyle\int_{\mathbb{R}}\nu(dz)>\Big|\frac{f(t)-r}{g(t)}\Big|$ $\forall t \in [0,\infty)$
% }
% \end{enumerate}

 \begin{theorem} \label{t.2.1} 
%We assume that the  L\'evy   process  $Z(t)$ has a bounded negative jumps. This means that the support 
%of the Poisson random measure is a bounded set $\JJ \subseteq [-R, \infty)$. 
Suppose that $f  :\RR\rightarrow \RR$  and
$g:  \JJ \times \RR\rightarrow \RR$ are  bounded measurable  functions such that  there is a 
constant $\al_0>1$ satisfying 
$g(z, x) \ge \al_0>-1$ 
%\footnote{{\color{blue}are we assuming $g>0$ }} 
for all $z\in \JJ$ and for all 
$ x\in \RR$,  where $\JJ$ is the supporting set of the Poisson measure $N(t, dz)$.  
%suppose that the  initial process $\phi(t)$ is a positive  measurable function. 
Then,  the stochastic differential delay equation 
\eqref{e.2.5}  admits a unique pathwise solution
 with the property that  %for all $t > 0$,   $X(t) \geq 0$   almost surely. I
 if   $\phi(0)> 0$,  then for all $t > 0$,  the random variable 
  $X(t) > 0$  almost surely.
  \end{theorem}
 \begin{proof} First, let us consider the interval $[0, b]$. 
When $t$ is in this interval $f(X(t-b))= f(\phi(t-b))$ 
and $g(z;X(t- b)) = g(z; \phi(t-b))$
%\deleted{$g(z, X(t-b))=g(\phi(t-b))$} \added{$g(z, X(t-b))=g(z, \phi(t-b))$}  
are known given functions of $t$ (and $z$). 
% \begin{eqnarray}
% dX(t) &=& X(t)\Big((f(X(t-b)) - \lambda Lg(X(t-b)) )dt +  g(X(t-b))dZ(t)\Big)\nonumber\\&&{\label{pema1}} 
% \end{eqnarray}
% with $X(t) = \phi(t)$ for $t \in [-R,0]$
Thus, \eqref{e.2.5} is a linear equation driven by Poisson random measure. 
The standard theory (see e.g. \cite{d.applebaum, protter}) can be used to show that the equation has a unique solution. Moreover, it is also well-known (see the above mentioned books or \cite{ahs})   that 
by It\^o's formula the solution to \eqref{e.2.5} can be written as 
 \begin{eqnarray*}
 X(t) &=&\phi(0)\exp\bigg\{ 
 \int_{0}^{t} f(\phi(s-b)) ds+ \int_{[0, t]\times \JJ}
   \log\left[  1+g(z, \phi (s-b) ) \right] \Tilde{N}(ds,dz)\\
  &&\qquad  +\int
  _{[0, t]\times \JJ} \Big( \log\left[1+
  g(z, \phi (s-b))\right] -g(z, \phi (s-b)) \Big)ds\nu(dz)     \bigg\}\,.
%&& \exp\Big( \displaystyle\int_{0}^{t}(f(\phi(u-b))  - {\lambda}Lg(\phi(u-b)))du\\&&+ \displaystyle\int_{0}^{t}g(\phi(u-b))dZ(u)\Big)\\
%&& \cdot\prod_{0 \leq u\leq t, \Delta Z(u) \neq 0 }\exp\Big(\ln(1+g(\phi(u-b)))Y_{N(u)}  - g(\phi(u-b))Y_{N_u}\Big)
 \end{eqnarray*}
From this formula we    see  that 
%$X(t)\geq 0$  almost surely 
%if $\phi(0)\geq 0 $ for $t \in [0, b]$
%and that 
if $\phi(0)>0$,  then the random variable $X(t)> 0$  almost surely 
  for every  $t \in [0, b]$.  
  
In similar way, we can consider the equation \eqref{e.2.5} 
on $t\in [kb, (k+1)b]$   recursively   for $k=1, 2, 3, \cdots$, 
 and obtain the same statements  on this interval from previous results on the interval $t\in [-b, kb]$.  
 \end{proof}
Since \eqref{e.2.1} is a special case of \eqref{e.2.5}, we can write down a corresponding result of  the  above theorem for 
\eqref{e.2.1}.

\begin{corollary} \label{c.2.2} 
Let the  L\'evy   process  $Z(t)$ have    bounded negative jumps
%$\JJ\subset. This means that the support 
%%of the Poisson random measure is a bounded set 
(e.g. $\De Z(t)\in \JJ \subseteq [-R, \infty)$). 
Suppose that $f, g:\RR\rightarrow \RR$ are  bounded measurable  functions such that  there is a 
constant $\al_0>1$ satisfying 
$g(  x) \le \frac{\al_0}{R} $ for all   $ x\in \RR$.
% and  
%suppose that the  initial process $\phi(t)$ is a positive  measurable function. 
Then,  the stochastic differential delay equation 
\eqref{e.2.1}  admits a unique pathwise solution
 with the property that % for all $t > 0$,   $X(t) \geq 0$   almost surely. I
 if   $\phi(0)> 0$,  then for all $t > 0$ the random variable  $X(t) > 0$  almost surely.
 \end{corollary}
\begin{proof} Equation \eqref{e.2.1} is a special case of \eqref{e.2.5} with $g(z,x)=z g(x)$.  The condition
$g(x) \le \frac{\al_0}{R} $ implies 
$g(z, x) \ge \al_0>-1$ 
%\footnote{\color{blue}} 
for all $z\in \JJ$ and for all 
$ x\in \RR$. Thus, Theorem \ref{t.2.1} can be applied. 
\end{proof}  
 
 \begin{example}
 One example of 
 the L\'evy process $Z(t)$ we have in mind which is 
 used in finance is the hyper-exponential jump process, which we explain below.  Let $Y_i, i=1, 2, \cdots$ be independent
and identically distributed random variables with the probability distribution given by 
\[
f_Y(x) = \sum_{i=1}^{m} p_i\eta_i e^{-\eta_i x}I_{\{x \geq 0\}}+ \sum_{j=1}^{n} %\deleted{q_i} 
 q_j \theta_j e^{\theta_j x}I_{\{x < 0\}}\,, 
\]
where   
\[
\eta_i>0,\  p_i \ge 0, \quad \theta_j>0, \  q_j\ge 0\,, \quad 
i=1, \cdots, m,\  j=1, \cdots, n
\]
 with $\sum_{i=1}^m p_i+\sum_{j=1}^n q_j =1$.  Let  $N_t$ be a 
 Poisson process with intensity $\la$.
 Then  
 \[
  {Z}(t) =   \sum_{i=1}^{N_t}Y_i 
  \]
is a L\'evy process. 
%
% Since its a L\'evy measure we will have $\displaystyle\int_{|z|\leq1}|z^2 \wedge 1| \nu (dz) < \infty$. 
% In our case (for DEJP and HEJP) we will have finite activity of the L\'evy measure, i.e $\nu (\mathbb{R}_0)<\infty$. 
% The discussion below will be valid if we can write $\displaystyle\int_{\mathbb{R}_0} \nu (dz) =\lambda \displaystyle\int_{\mathbb{R}_0}f_Y(z)dz$ and we have $\nu (\mathbb{R}_0)<\infty$. where $f_Y$ is the density of the jump size.  \newline 
% $Y_i$ is the jump sizes. We discuss the cases where the jump sizes follows the DEJP or HEJP.
% \begin{itemize}
%     \item  
%  
 If  $m=1, n=1$ then $Z(t)$ is called a double exponential process.
%  
%The class of Hyper-exponential processe is rich enough to approximate many heavy tailed distribution, power tailed distribution in the sense of weak distribution. It is flexible enough to incorporate the uncertainty of the heavines of the asset return tails therefore can capture the leptokurtic feature. (leptokurtic feature = fat tails + kurtosis )
%   
The assumption on the boundedness of the  negative jumps can be made possible by requiring that  $q_j=0$ for   all $j=1, \cdots, n$ or by replacing the negative exponential distribution by truncated negative exponential 
distributions, namely, 
\[
f_Y(x) = \sum_{i=1}^{m} p_i\eta_i e^{-\eta_i x}I_{\{x \geq 0\}}+ \sum_{j=1}^{n} q_j\frac{\theta_j}{1-e^{-\theta_j R_j}} e^{\theta_j x}I_{\{-R_j <x < 0\}}\,, 
\]
where   
\[
\eta_i>0,\  p_i \ge 0, \quad \theta_j>0, R_j>0, \  q_j\ge 0\,, \quad 
i=1, \cdots, m,\  j=1, \cdots, n
\]
 with $\sum_{i=1}^m p_i+\sum_{j=1}^n q_j =1$.  For this
 truncated hyper-exponential   process, we can take $
 \JJ=[-R, \infty)$ with $R=\max\{R_1, \cdots, R_n\}$. 
\end{example}
Although this paper will mainly concern with the one dimensional 
delayed  stochastic differential equation \eqref{e.2.5} or \eqref{e.2.1} it is interesting to extend  Theorem \ref{t.2.1} to more than one dimension. 

Let $\tilde N_j(ds,dz)$, $j=1, \cdots,  d$
be independent compensated Poisson random measures. Consider the following system of delayed stochastic differential equations driven by Poisson random measures:
\begin{eqnarray}
dS_i(t)&=&   \sum_{j=1}^d f_{ij} (S(t-b))  
S_j (t) dt \nonumber\\
&&\quad +S_i(t-)  \sum_{j=1}^d \int_{\JJ} g_{ij} (z, S(t-b))   \tilde N_{j} (dt, dz) \,, \quad i=1, \cdots, d\,,
\nonumber\\
S_i(t) &=& \phi_i(t) \,,\quad t\in [-b, 0]\,, \ i=1, \cdots, d\,,   
\label{e.2.6} 
\end{eqnarray} 
where  $S(t)=(S_1(t), \cdots, S_d(t))^T$. 
 \begin{theorem} \label{t.2.4}  
Suppose that $f_{ij}
:\RR \rightarrow \RR$  and $  g_{ij}:\JJ\times \RR\rightarrow \RR\,, \ 
1\le i,j\le d$ are  bounded measurable  functions such that  there is a 
constant $\al_0>1$ satisfying 
$g_{ij}(z, x) \ge \al_0>-1$ for all $
1\le i,j\le d$, for all $z\in \JJ$ and for all 
$ x\in \RR$,  where $\JJ$ is the common supporting set of the Poisson measures $\tilde N_{j}(t, dz), j=1, \cdots, d$.  
If for all $i\not= j$, $f_{ij}(x)\ge 0$ for all $x\in \RR$, 
and $\phi_i(0)\ge 0\,, \ i=1, \cdots, d$, then,  the stochastic differential delay equation 
\eqref{e.2.6}  admits a unique pathwise solution
 with the property that    for all $i=1, \cdots, d$  and for all $t > 0$,
    the random variable  $S_i(t) \geq 0$   almost surely.  
  \end{theorem}  
\begin{proof} We can follow the argument as in the proof of Theorem \ref{t.2.1} to show that   the system of delayed stochastic differential 
equations \eqref{e.2.6} has a unique solution $S(t)=(S_1(t), \cdots, S_d(t))^T$. We shall modify slightly  the method of \cite{humulti} to show 
the positivity of the solution.  
Denote $\tilde g_{ij}(t,z)=g_{ij}(z,S(t-b))$.   
Let $Y_i(t)$ be the solution to the stochastic differential equation
\[
dY_i(t) =     Y_i(t-)  \sum_{j=1}^d \int_{\JJ} \tilde g_{ij} (t,z)   \tilde N_{j} (dt, dz) 
\]
with initial conditions $Y_i(0)=\phi_i(0)$.  Since this is a scalar equation for $Y_i(t)$, its explicit solution can be represented 
\begin{eqnarray*}
 Y_i(t) &=&\phi_i(0)\exp\bigg\{  
 \sum_{j=1}^d    \log\left[  1+\tilde g_{ij} 
 (s,z ) \right] \Tilde{N}_j(ds,dz)\\
  &&\qquad  + \sum_{j=1}^d \int
  _{[0, t]\times \JJ} \Big( \log\left[1+\tilde 
  g_{ij}(s,z)\right] -\tilde g_{ij}(s, z ) \Big)ds\nu_j(dz)     \bigg\}\,,  
 \end{eqnarray*}
 where $\nu_j$ is the associated L\'evy measure for $\tilde N_j(ds, dz)$.  Denote $\tilde f_{ij}(t)=f_{ij}(S(t-b))$ and  let 
$p_i(t)$ be the solution to the following system of 
equations
\[
dp_i(t) =\sum_{j=1}^d \tilde f_{ij}(t)p_j(t) dt\,, \quad 
p_i(0)=1\,, \quad i=1, \cdots, d\,.  
\]
By the assumption on $f$ we have that when $i\not=j$, $\tilde 
f_{ij}(t)\ge 0$ almost surely. By a theorem in \cite[p.173]{Be60} we see
that
$p_{i}(t)\ge 0$ for all $t\ge 0$  almost surely. 
Now it is easy to check by the It\^o formula that $\tilde S_i(t)=p_i(t)Y_i(t)$  is the solution to
\eqref{e.2.6} which satisfies that 
$\tilde S_i(t)\ge 0$ almost surely. By the uniqueness of the solution 
we see that $S_i(t)=\tilde S_i(t)$ for $i=1, \cdots, d$.  The theorem is then proved. 
\end{proof}
  
\section{%Positivity preserving  l
Logarithmic 
Euler-Maruyama scheme}\label{s.3} 
The equation \eqref{e.2.1} or \eqref{e.2.5} is used in Section
\ref{s.4}  
to model the price of a  risky  asset in a  financial market
and its the solution 
is proved to be  positive as in Theorem \ref{t.2.1}. 
As it is well-known the usual Euler-Maruyama scheme cannot preserve 
the positivity of the solution (e.g. \cite{huyi} and references therein). 
Motivated by the work \cite{huyi}, 
  we propose in this section a variant of the  Euler-Maruyama scheme 
  (which we call logarithmic 
Euler-Maruyama scheme)  to approximate the solution so that all   approximations 
 are always non-negative.  For the convenience of the future 
simulation, we shall consider only the 
equation \eqref{e.2.1}, which we rewrite here:   
 \begin{eqnarray}
 dS(t) &=& f(S(t-b))S(t)dt +   g(S(t-b)) S(t-)d {Z}(t)\,, 
 \label{e.3.1}
% \&& 
% =X(t)
% \Big((f(X(t-b)) - {\lambda}Lg(X(t-b)) )dt +  g(X(t-b))dZ(t)\Big)\nonumber \\&&{\label{pema101}}.
 \end{eqnarray}
% where $\tilde{Z}(t) = Z(t) - {\lambda}{L}t $ is a martingale and ${\lambda}t$ is the parameter of poisson process with $L = \displaystyle\int_{\mathbb{R}_0}z\nu(dz)$.
where 
$
Z(t)=     \sum_{i=1}^{N_t}Y_i  
$ is a L\'evy process.  
Here   $N_t$ is  a 
 Poisson process with intensity $\la$ and $Y_1, Y_2, \cdots,
 $ are iid random variables.

The solution to the above equation can be written as  
\begin{eqnarray}
S(t) 
%%&=&\phi(0)\exp\Big( \displaystyle\int_{0}^{t} f(X(u-b)  - \lambda Lg(X(u-b)))du\\&&+ \displaystyle\int_{0}^{t}g(X(u-b))dZ(u)\Big) \\&&\cdot\prod_{0 \leq u\leq t, \Delta Z(u) \neq 0 }\exp\Big(\ln(1+g(X(u-b))Y_{N(u)})  - g(X(u-b))Y_{N_u}\Big)
 =\phi(0)\exp\Big( \displaystyle\int_{0}^{t} f(X(u-b))
 %-\lambda Lg(X(u-b)))
 du+\sum_{0\leq  u\leq t, \Delta Z(u)\neq 0}\ln(1+g(X(u-b))Y_{N(u)})  \Big)\,. \label{e.3.2} 
 \end{eqnarray}

%Let the time step size  $\Delta \in (0,1)$ be 
% \begin{eqnarray*}
% \Delta = \frac{R}{N}
% \end{eqnarray*}
% for sufficiently large $N$. 
% Based on above discussion we define the continuous Euler-Maruyama approximate solution

 We shall consider a finite  time 
interval $[0, T]$ for some fixed $T>0$. 
Let $\Delta =\frac{T}{n}>0$ be a time step size for some positive integer $n\in \NN$.  For any nonnegative  integer $k\ge 0$, denote $t_k=k\De$.  
We consider the partition $\pi$   of the time interval  $[0, T]$:
\[
\pi: 0=t_0<t_1<\cdots  <t_n=T\,.
\] 
On the subinterval $[t_k,  t_{k+1}]$ the solution \eqref{e.3.2} can also be written as 
 \begin{eqnarray}
S(t)  
 &=& S(t_k) \exp\Big( \displaystyle\int_{t_k }^{t}  f(X(u-b) ) 
 du\nonumber\\
 &&\qquad  +\sum_{t_k\leq  u\leq t, \Delta Z(u)\neq 0}\ln(1+g(X(u-b))Y_{N(u)})  \Big)\,, t\in [t_k, t_{k+1}]\,. \label{e.3.6} 
 \end{eqnarray}
Motivated by the formula  \eqref{e.3.6}, 
we propose  a  logarithmic Euler-Maruyama scheme  to approximate 
\eqref{e.2.1} as follows. 
 \begin{eqnarray} 
 S^\pi (t_{k+1} ) 
 &=&S^\pi (t_k )\exp\Big( f(S^\pi (t_k -b))  \Delta\Big)\nonumber\\
 && \quad  \cdot\exp\Big( \ln(1+g(S^\pi (t_k -b)) \Delta Z_k)\Big)\,,  \hspace{3mm} k=0,1,2,..., n-1  \nonumber\\   \label{e.3.4} 
 \end{eqnarray}
 with $S^\pi (t) = \phi(t) $  for all  $ t \in [-b,0]$.
It is clear that if  $\phi(0)>0$,  then 
$S^\pi(t_k)>0$ almost surely for all $k=0,1,2,..., n$.  Then our approximations $S^\pi(t_k) $  are always positive. Notice that the 
approximations from usual Euler-Maruyama scheme 
is always not positive preserving (see e.g. \cite{huyi}
and references therein). 

We shall prove the convergence and find the rate of convergence for the above scheme.  For the convergence of the usual Euler-Maruyama  scheme of jump  equation
with  delay, we refer to  \cite{convg}.  To study the convergence of the above logarithmic Euler-Maruyama scheme, we  make the following assumptions. 
 \begin{enumerate}
 \item[{\bf (A1)}] The initial data $\phi(0)>0$ and it is H\"older continuous i.e there exist constant $\rho >0$ and $\gamma \in [1/2,1])$ such that for $t,s \in [-b,0]$
 \begin{eqnarray}
 |\phi(t) - \phi(s)| \leq \rho |t-s|^{\gamma}.
 \end{eqnarray}
% 
% 
%  \item[A2] $f,g$ are sufficiently smooth.
% 
 \item[{\bf (A2)}] $f$ is bounded. $f$ and $g$ are  global Lipschitz. This means that there exists a     constant  $\rho>0$ such that
 \begin{eqnarray}
\begin{cases}\Big|g(x_1)-g(x_2)\Big| 
 \leq \rho |x_1 - x_2| \,;\\ 
 \Big|f(x_1)-f(x_2)\Big| 
 \leq \rho  |x_1 - x_2|\,\,,\quad \forall \ x_, x_2\in \RR\,;
 \nonumber \\
  \big|f(x)\big|  \leq \rho  \,,\quad \forall x\in \RR
 \end{cases} \label{pema102}
 \end{eqnarray} 
 \item[{\bf (A3)}] The support  $\JJ $ of   the Poisson random measure $N$ is contained in $[-R, \infty)$  for some $R>0$ and   there are   
constants  $\al_0>1$  and $\rho>0$ satisfying 
$-\rho\le g(  x) \le \frac{\al_0}{R} $ for all   $ x\in \RR$. 
 \item[{\bf (A4)}] For any $q>1$  there is a $\rho_q>0$ 
 \begin{eqnarray}
 \displaystyle\int_{\JJ}(1+ |z |)^q\nu(dz) \leq \rho_q \,,  \hspace{3mm}\forall x  \in \mathbb{R}\,.  \label{pema103}
 \end{eqnarray}
 \end{enumerate} 
For notational simplicity we  introduce two step processes
 \begin{eqnarray*}
\begin{cases} v_1(t) = \sum_{k=0}^{\infty}\111_{[t_k , t_{k+1}
%\deleted{]}\added{)}
)}(t)S^\pi (t_k )\\  
 v_2(t) = \sum_{k=0}^{\infty}\111_{[t_k , t_{k+1}
 %\deleted{]}\added{)}
 )}(t)S^\pi (t_k -b ).
 \end{cases} 
 \end{eqnarray*}
Define the continuous interpolation of the logarithmic 
Euler-Maruyama approximate solution on the whole interval $[-b, T]$ 
(not only on $t_k, k=0, \cdots, n$)  as follows: 
 \begin{eqnarray}
S^\pi(t) =
     \begin{cases}
       \phi(t) & \qquad t \in [-b, 0]\\
       \phi(0)\exp\Big(\displaystyle\int_{0}^{t} f(v_2(u)) du\\
     \qquad\qquad   + \sum_{0\leq  u\leq t, \Delta Z(u)\neq 0}\ln(1+g(v_2(u))  Y_{N(u)}) \Big)    &\qquad  t \in [0,T].
     \end{cases} \label{pema105}       
 \end{eqnarray}
 With this %\deleted{interploation} \added{interpolation}
 interpolation, we see that $S^\pi(t)>0$ almost surely for all $t\ge 0$. 
 
 \begin{lemma}\label{l.3.1} 
 Let the assumptions (A1)-(A4) be satisfied. Then for any $q\ge 1$ there exists $K_q$, independent of the partition $\pi$, 
  such that
  \begin{eqnarray*}
  \mathbb{E}\Big[\sup_{0 \leq t \leq T}|S(t)|^q \Big] \vee \mathbb{E}\Big[\sup_{0 \leq t \leq T}|S^\pi(t) |^q \Big]
 \leq K_q.
 \end{eqnarray*}
 \end{lemma}
 \begin{proof}  We can assume that $q>2$.   First,   let us prove
 $   \mathbb{E}\Big[\sup_{0 \leq t \leq T}|S^\pi(t) |^q \Big]
 \leq K_q $.  
From \eqref{pema105} it follows 
 \begin{eqnarray*} 
 \mathbb{E}\Big[\sup_{0 \leq t \leq T}|S^\pi(t) |^q \Big] 
 &\leq & |\phi(0)|^q \mathbb{E}\Big[\sup_{0 \leq t \leq T}\exp\Big(q\displaystyle\int_{0}^{t} f(v_2(u)) du\\
&&\qquad  + q\sum_{0\leq  u\leq t, \Delta Z(u)\neq 0}\ln(1+g(v_2(u))Y_{N(u)}) \Big)\Big]\,. 
 \end{eqnarray*}
Since   $|f(t)|\le \rho$  we have 
 \begin{eqnarray}
 &&
 \mathbb{E}\Big[\sup_{0 \leq t \leq T}|S^\pi(t) |^q \Big] 
 \nonumber \\
 &&\leq \phi(0)^q e^{q\rho T  }
 \mathbb{E}\Big[\sup_{0 \leq t \leq T}  \exp\Big( q\sum_{0\leq  u\leq t, \Delta Z(u)\neq 0}\ln(1+g(v_2(u)) Y_{N(u)}) \Big)\Big]
  \nonumber\\
  && = \phi(0)^q e^{q\rho T  } \mathbb{E}\Big[\sup_{0 \leq t \leq T}\exp\Big(  q\displaystyle\int_{  \mathbb{T} }\ln(1+zg(v_2(u)))N(du,dz)\Big)\Big]\,,\label{e.3.8} 
% \\&&\leq \phi(0)\exp\Big(2\lambda qLBT+qrT\Big)\mathbb{E}\Big[\sup_{0 \leq t \leq T}\exp\Big(q\displaystyle\int_{  \mathbb{T} }\ln(1+zg(v_2(u)))N(du,dz)\Big)\Big]
 \end{eqnarray}
 where and throughout the remaining part of this paper, we denote
 $\TT=[0, t]\times \JJ$. 
Now we are going to handle    the factor
\[
I:=\mathbb{E}\Big[\sup_{0 \leq t \leq T}\exp\Big(q\displaystyle\int_{  \mathbb{T} }\ln(1+zg(v_2(u)))N(du,dz)\Big)\Big] \,.
\]
Let $h = ((1+zg(v_2(u))^{2q}-1))/z  $.  Then
 \begin{eqnarray*}
 I&=&
 \mathbb{E}\Big[\sup_{0 \leq t \leq T}\exp\Big(\frac12 \displaystyle\int_{  \mathbb{T} }\ln(1+zh)N(du,dz) \Big)\Big] \\
 &=&
 \mathbb{E}\Big[\sup_{0 \leq t \leq T}\exp\Big(\frac12 \displaystyle\int_{  \mathbb{T} }\ln(1+zh)\tilde N(du,dz) 
 +\frac12  \int_{  \mathbb{T} }\ln(1+zh)\nu(dz) du 
   \Big)\Big] \\
   &=&
 \mathbb{E}\Big[\sup_{0 \leq t \leq T}\exp\Big(\frac12 \displaystyle\int_{  \mathbb{T} }\ln(1+zh)\tilde N(du,dz) 
 +\frac12  \int_{  \mathbb{T} }\left[ \ln(1+zh)-zh\right] 
 \nu(dz) du 
   \Big)\Big] \\
 &&\qquad    \sup_{0 \leq t \leq T}\exp\Big( 
 -\frac12  \int_{  \mathbb{T} }  (1+zg(v_2(u))^{2q}-1)\ 
 \nu(dz) du 
   \Big)\Big] \\
    &\le &C_q
 \mathbb{E}\Big[\sup_{0 \leq t \leq T}\exp\Big(\frac12 \displaystyle\int_{  \mathbb{T} }\ln(1+zh)\tilde N(du,dz) 
 +\frac12  \int_{  \mathbb{T} }\left[ \ln(1+zh)-zh\right] 
 \nu(dz) du 
   \Big)\Big]  \,, 
% &=&\mathbb{E}\Big[\sup_{0 \leq t \leq T}\exp\Big(2\displaystyle\int_{  \mathbb{T} }zh\nu(dz)du \Big)\\&& \cdot\exp\Big(2\displaystyle\int_{  \mathbb{T} }\ln(1+zh)\tilde{N}(du,dz) +  2\displaystyle\int_{  \mathbb{T} }(\ln(1+zh) -zh)\nu(dz)du \Big)\Big]
 \end{eqnarray*}
where we used  boundedness of $ g   $  and the assumption (A4).   Now an application of the Cauchy-Schwartz  inequality 
yields 
\begin{eqnarray*}
 I
 &\le &   C_q
 \bigg\{\mathbb{E}\Big[\sup_{0 \leq t \leq T}M_t \Big]\bigg\}^{1/2}   \,, 
 \end{eqnarray*} 
 where 
 \[
 M_t:=\exp\Big(  \displaystyle\int_{  \mathbb{T} }\ln(1+zh)\tilde N(du,dz) 
 +   \int_{  \mathbb{T} }\left[ \ln(1+zh)-zh\right] 
 \nu(dz) du 
   \Big)
   \,.
   \]
But $(M_t, 0\le t\le T)$ is an exponential  martingale.    
Thus,
\[
\mathbb{E}\Big[\sup_{0 \leq t \leq T}M_t \Big]\le 2
\mathbb{E}\Big[ M_T \Big]=2\,.
\]
Inserting this estimate  of  $I$ into 
\eqref{e.3.8} proves 
$\mathbb{E}\Big[\sup_{0 \leq t \leq T}|S^\pi(t) |^q \Big] \le
K_q<\infty$. In the same way we can show 
$\mathbb{E}\Big[\sup_{0 \leq t \leq T}|S (t) |^q \Big] \le
K_q<\infty$.  This completes the proof of the lemma. 
% \begin{eqnarray*}&&
% \mathbb{E}\Big[\sup_{0 \leq t \leq T}\exp\Big(2\displaystyle\int_{  \mathbb{T} }\ln(1+zh)N(du,dz) \Big)\Big] \leq \exp(L2^q )\\&&\cdot\mathbb{E}\Big[\sup_{0 \leq t \leq T}\exp\Big(2\displaystyle\int_{  \mathbb{T} }\ln(1+zh)\tilde{N}(du,dz) +  2\displaystyle\int_{  \mathbb{T} }(\ln(1+zh) -zh)\nu(dz)du \Big)\Big]
% \end{eqnarray*}
% We apply the Burkholder davis gundy inequality. We note that using the Ito formula and the $$[\bar{X}(t),\bar{X}(t)] = 1+ \int_{  \mathbb{T}  }z^2h^2\exp(\bar{X}(s))N(dz,ds)$$
% where $\bar{X}(t)=\exp\Big(\displaystyle\int_{  \mathbb{T} }\ln(1+zh)\tilde{N}(du,dz) +  \displaystyle\int_{  \mathbb{T} }(\ln(1+zh) -zh)\nu(dz)du \Big)$
% hence 
% \begin{eqnarray*}&&
% \mathbb{E}\Big[\sup_{0 \leq t \leq T}\exp\Big(\displaystyle\int_{  \mathbb{T} }\ln(1+zh)N(du,dz) \Big)\Big] \leq \exp(LB2^q )\\&&\cdot\Big(1+ \mathbb{E}\Big[\Big(\int_{  \mathbb{T}  }z^2h^2\exp(\bar{X}(s))N(dz,ds)\Big)\Big]\Big)
% \end{eqnarray*}
% we estimate {\color{red}  
% \[
% \mathbb{E}\Big[\displaystyle\int_{  \mathbb{T}  }z^2h^2\exp(\bar{X}(s))N(dz,ds)\Big]
% = \mathbb{E}\Big[\displaystyle\int_{  \mathbb{T}  }z^2h^2\exp(\bar{X}(s))\nu(dz)ds \Big]
% \]
% \[
%   \displaystyle =\int_{  \mathbb{T}  }z^2h^2\mathbb{E}\Big[ \exp(\bar{X}(s))  
%   \big]\nu(dz)ds 
%  \]
%  
% } 
 \end{proof}

 %\begin{eqnarray*}
 %\mathbb{E}_{\mathbb{Q}}\Big[\sup_{0 \leq t \leq T}|X(t)-S^\pi(t)  |^2 \Big] \leq 
 %\mathbb{E}_{\mathbb{Q}}\Big[\sup_{0 \leq t \leq T}|S^\pi(t) |^2|\frac{X(t)}{S^\pi(t) }-1 |^2 \Big]
 %\end{eqnarray*}
 
 \begin{lemma} \label{l.3.2}
Assume  (A1)-(A4).  Then there is a constant $K>0$,
independent of $\pi$,  such that 
 \begin{eqnarray*}
 \mathbb{E}_{\mathbb{Q}}\Big|S^\pi(t)- v_1(t) \Big|^{p} \leq K \Delta^{p/2},\hspace{4mm} \forall\ t \in [0,T]\,. 
 \end{eqnarray*}
% where $K_3$ is independent of $\Delta$.
 \end{lemma}
 \begin{proof}
Let  $t\in [t_j,t_{j+1})$ for some $j $. Using $
|e^x-e^y|\le (e^x+e^y)|x-y|$  
 we can write 
 \begin{eqnarray*} 
 \Big| S^\pi(t) -v_1(t)\Big| 
 &=& \Big|S^\pi(t)-S^\pi(t_j)\Big|\\
 & 
 \leq&  \Big| S^\pi(t)+S^\pi(t_j)\Big|  \cdot \Big|\displaystyle\int_{t_j }^{t}f(v_2(s)) ds   + \sum_{t_j \leq s \leq t}\ln(1+g(v_2(s)) Y_{N(s)})   \Big|\,. 
 \end{eqnarray*}
An application of the H\"older inequality  yields that for any $p>1$, 
\begin{eqnarray} 
\EE  \left[\Big| S^\pi(t) -v_1(t)\Big| ^p\right]
 &   
 \leq&  \left\{\EE  \left[ \Big| S^\pi(t)+S^\pi(t_j)\Big|  \cdot \Big|\right]^{2p}\right \}^{1/2} \nonumber \\
 &&\qquad 
 \left\{\EE\left|\int_{t_j }^{t}f(v_2(s)) ds   + \sum_{t_j \leq s \leq t}\ln(1+g(v_2(s)) Y_{N(s)})  \right|^{2p} \right\}^{1/2}
 \nonumber \\
 &   
 \leq&  K_p 
 \left\{\EE\left|\int_{t_j }^{t}f(v_2(s)) ds\right|^{2p}    +
   \EE\left|  \sum_{t_j \leq s \leq t}\ln(1+g(v_2(s))Y_{N(s)})   \right|^{2p} \right\}^{1/2} \nonumber \\
&  \leq&  K_p 
 \left\{\De^{2p}     +
   \EE\left|  \sum_{t_j \leq s \leq t}\ln(1+g(v_2(s))Y_{N(s)})  \right|^{2p} \right\}^{1/2}\,.  \label{e.3.9}  
 \end{eqnarray}
% \begin{eqnarray*}&&
% \mathbb{E}\Big[\Big|m(u-b) - m(t_j-b)\Big|^p\Big]\leq \mathbb{E}\Big[ \Big|m(u-b)+m(t_j-b)\Big|^p\\&& \cdot \Big|\displaystyle\int_{t_j}^{u}f(v_2(u)) ds +- \lambda L\displaystyle\int_{t_j }^{u}g(v_2(s))ds + \sum_{t_j \leq s \leq u}\ln(1+g(v_2(s))Y_{N(s)}) \Big| ^p \Big].
% \end{eqnarray*}
% Since we have $$\tilde{N}(dz,ds)=N(dz,ds)-\nu(dz)ds $$ we can write
% 
Now we want to bound 
\[
I:=\EE\left|  \sum_{t_j \leq s \leq t}\ln(1+g(v_2(s))Y_{N(s)})  \right|^{2p}  \,.
\]
(we use the same 
notation $I$ to denote different quantities  in different 
occasions and this will not cause ambiguity). 
We write the above sum as an integral: 
 \begin{eqnarray*}
I
&=&\mathbb{E}\Big|\displaystyle\int_{\JJ }\displaystyle\int_{t_j} ^t\ln(1+zg(v_2(s))) {N}(ds,dz)\Big|^{2p}\\ 
&=&\mathbb{E}\Big|\displaystyle\int_{\JJ }\displaystyle\int_{t_j} ^t\ln(1+zg(v_2(s)))\tilde{N}(ds,dz)\\
&&+\displaystyle\int_{\JJ }\displaystyle\int_{t_j} ^t\ln(1+zg(v_2(s)))\nu(dz)ds\Big|^{2p}\\
&\le& C_p \left(\De^{2p} + \mathbb{E}\Big|\displaystyle\int_{\JJ }\displaystyle\int_{t_j} ^t\ln(1+zg(v_2(s)))\tilde{N}(ds,dz)  \Big|^{2p}\right)\,. 
 \end{eqnarray*}
 By the Burkholder-Davis-Gundy inequality, we have 
\begin{eqnarray*}
&& \mathbb{E}\Big|\displaystyle\int_{\JJ }\displaystyle\int_{t_j} ^t\ln(1+zg(v_2(s)))\tilde{N}(ds,dz)  \Big|^{2p}\\
 &  &\qquad \quad \le \mathbb{E}\left(\displaystyle  \int_{\JJ } \int_{t_j} ^t
   \Big| \ln(1+zg(v_2(s))) \Big|^2 \nu(dz)ds   \right)^p\\
 &  &\qquad \quad \le  K_p \De^p\,.
 \end{eqnarray*} 
 Thus, we have
 \[
 I\le K_{p,T}  \De^p\,.
 \] 
%  Hence we have
% \begin{eqnarray*}
% && \mathbb{E}\Big[\Big|m(u-b) - m(t_j-b)\Big|^p\Big] \leq 
%    2K_p^{.5}\Bigg(2^{(p-1)}\mathbb{E}\Big[ \Big|\displaystyle\int_{t_j}^{u}\Big|f(v_2(s))\Big|^{2} ds\Big|^p\Big] \\&&+ 2^{2(p-1)}\mathbb{E} \Big[\Big|\displaystyle\int_{t_j}^{u}\displaystyle\int_{\mathbb{R}_0}g(z,v_2(s))\tilde{N} (ds,dz)\Big|^{2p}\Big] \\&&+ 2^{4(p-1)}(\lambda L)^{2p}\mathbb{E} \Big[\Big|\displaystyle\int_{t_j}^{u}g(v_2(s))ds\Big|^{2p}\Big]\Bigg)^{.5}
% \end{eqnarray*}
% We now use that Burkholder Davis Gundy inequality, $\ln(1+x)<x $. By A4 we write
% \begin{eqnarray*}&& \leq 2K_p^{.5}\Bigg(2^{(p-1)}\mathbb{E}\Big[\Big|\displaystyle\int_{t_j}^{u} h(1+m^2(t_j -b))ds\Big|^{p}\Big] + 2^{2(p-1)}\mathbb{E}\Big[\Big|\displaystyle\int_{t_j}^{u} h(1+m^2(t_j -b))ds\Big|^{p}\Big] \\&&+ 2^{4(p-1)}\mathbb{E}\Big[\Big|\displaystyle\int_{t_j}^{u} h(1+m^2(t_j -b))ds\Big|^{p}\Big]\Bigg)^{.5}\\&&
% \leq 
% 2K_p^{.5}\Bigg(2^{(p-1)}\mathbb{E}\Big[\Big| h(1+m^2(t_j -b))\Delta \Big|^{p}\Big] + 2^{2(p-1)}\mathbb{E}\Big[\Big|h(1+m^2(t_j -b))\De\Big|^{p}\Big] \\&&+ 2^{4(p-1)}\mathbb{E}\Big[\Big| h(1+m^2(t_j -b))\De\Big|^{p}\Big]\Bigg)^{.5}\\&&
% \leq 
% 2K_p^{.5}\Bigg(2^{(p-1)}.2^{p-1}( h^p\De^p+h^p\De^pK_{2p}) + 2^{2(p-1)}2^{p-1}( h^p\De^p+h^p\De^pK_{2p}) \\&&+ 2^{4(p-1)}2^{p-1}( h^p\De^p+h^p\De^pK_{2p})\Bigg)^{.5}:=K_3\De^{p/2}
% \end{eqnarray*}
% where \begin{eqnarray*}
% K_3&:=& 2K_p^{.5}\Bigg(2^{(p-1)}.2^{p-1}( h^p+h^pK_{2p}) + 2^{2(p-1)}2^{p-1}( h^p+h^pK_{2p}) \\&&+ 2^{4(p-1)}2^{p-1}( h^p+h^pK_{2p})\Bigg)^{.5}
% \end{eqnarray*}
Inserting this bound into \eqref{e.3.9} yields the lemma.  \end{proof}
  Our next objective is to obtain the  rate of convergence 
  of our logarithmic Euler-Maruyama approximation 
  $S^\pi(t)$ to   the true solution $S(t)$. 
%  show that \begin{eqnarray*}
% \lim_{\Delta \rightarrow 0 }\mathbb{E}_{\mathbb{Q}}\Big[\sup_{0 \leq t \leq T}|X(t)-S^\pi(t)  |^p \Big] = 0.
% \end{eqnarray*}
  \begin{theorem}{\label{pema108}}
Assume  (A1)-(A4).   Let $S^\pi(t)$ be the solution 
to \eqref{e.3.4} and let $S(t)$ be 
the solution to \eqref{e.3.1}.  Then there is a constant $K_{p, T}$, independent of $\pi$ such that 
 \begin{eqnarray}
 \mathbb{E}_{\mathbb{Q}}\Big[\sup_{0 \leq t \leq T}|S(t)-S^\pi(t)  |^p \Big] \le K_{p, T} \De^{p/2}\,. 
 \end{eqnarray}
  \end{theorem}
 
 \begin{proof}
 We write $S(t)=\phi(0)\exp{(X(t))}$ and $S^\pi(t) =\phi(0)\exp{(p(t))}$.  Then 
 \begin{eqnarray*}
 && \Big|S(t)-S^\pi(t) \Big|^p \leq  \Big|S(t)+S^\pi(t) \Big|^p\Big|X(t)-p(t)\Big|^p \,. 
 \end{eqnarray*}
 Hence  by Lemma \ref{l.3.1}  we   have   for any $r\in [0, T]$
 \begin{eqnarray}&&
 \mathbb{E} \Big[\sup_{0 \leq t \leq  r}
 |S(t)-S^\pi(t)  |^p \Big] 
% \leq \mathbb{E}_{\mathbb{Q}}\Big[\sup_{0 \leq t \leq r}\Big|X(t)+S^\pi(t) \Big|^p\Big|S(t)-p(t)\Big|^p\Big]
\nonumber  \\
 &&\leq \mathbb{E} \Big[\sup_{0 \leq t \leq r}\Big|S(t)+S^\pi(t) \Big|^{2p}\Big]^{1/2}\mathbb{E} \Big[\sup_{0 \leq t \leq r}\Big|X(t)-p(t)\Big|^{2p}\Big]^{1/2}
\nonumber  \\
 && \leq  2^{2p-1}\Big(\mathbb{E} \Big[\sup_{0 \leq t \leq r}\Big|S(t)\Big|^{2p}\Big]+\mathbb{E} \Big[\sup_{0 \leq t \leq r}\Big|S^\pi(t) \Big|^{2p}\Big]\Big)^{1/2}\Big[ \mathbb{E}   \sup_{0 \leq t \leq r}\Big|X(t)-p(t)\Big|^{2p} \Big]^{1/2}
\nonumber   \\
&& \leq K_p \Big[ \mathbb{E} \sup_{0 \leq t \leq r}\Big|X(t)-p(t)\Big|^{2p}\Big]^{1/2}
 =K_p I^{1/2} \,. \label{e.3.10a} 
 \end{eqnarray}
Thus  we need only to bound the above expectation $I$, which is given by the following.   
 \begin{eqnarray} 
 I&=& \mathbb{E}\Big[\sup_{0 \leq t \leq r}|X(t)-p(t) |^{2p} \Big]  \nonumber\\
 &&\le \mathbb{E}\sup_{0 \leq t \leq r}\Big| \displaystyle\int_{0}^{t}(f(S(u-b))-f(v_2(u)))du  \label{e.3.10} \\
 &&+ \sum_{0\leq  u\leq t, \Delta Z(u)\neq 0}\ln(1+g(S(u-b))Y_{N(u)})-\ln(1+g(v_2(u))Y_{N(u)})\Big|^{2p}\,. 
 \nonumber
 \end{eqnarray}
By the Lipschitz conditions  we   have
 \begin{eqnarray} 
 I&& \leq K_p \mathbb{E} \displaystyle\int_{0}^{r}\Big|S(u-b)-v_2(u)\Big|^{2p}du \nonumber\\
 &&+ K_p 
 \mathbb{E}\sup_{0 \leq t \leq r}\Big|\sum_{0\leq  u\leq t, \Delta Z(u)\neq 0}\ln(1+g(S(u-b)) Y_{N(u)})-\ln(1+g(v_2(u))Y_{N(u)})\Big|^{2p}\nonumber 
 \\
 && \leq K_p \Big[\mathbb{E} \displaystyle\int_{0}^{r}\Big|S(u-b)-S^\pi(u-b)\Big|^{2p}du+\mathbb{E} \displaystyle\int_{0}^{r}\Big|S^\pi(u-b)-v_2(u)\Big|^{2p}du\Big] \nonumber\\
% &&+ 2^{2(p-1)}(\lambda L)^{2p}\Big[\mathbb{E} \displaystyle\int_{0}^{T}\Big|X(u-b)-m(u-b)\Big|^{2p}du+\mathbb{E} \displaystyle\int_{0}^{T}\Big|m(u-b)-v_2(u)\Big|^{2p}du\Big]\nonumber\\
 &&+K_p \mathbb{E}\sup_{0 \leq t \leq r}\Big|\sum_{0\leq  u\leq t, \Delta Z(u)\neq 0}\ln(1+g(S(u-b)) Y_{N(u)})-\ln(1+g(v_2(u))Y_{N(u)})\Big|^{2p}\nonumber\\
 &&=I_1+I_2+I_3\,.   \label{e.3.11} 
 \end{eqnarray}
%since we have $$\tilde{N}(dz,ds)=N(dz,ds)-\nu(dz)ds $$ 
By Lemma \ref{l.3.2}  and  by the assumption (A1) about the H\"older continuity of 
the initial data $\phi$ we have
\begin{eqnarray}
I_2\le K_{p, T} \De^p\,. \label{e.3.12} 
\end{eqnarray}
We   write the above sum $I_3$  with jumps as a stochastic integral: 
 \begin{eqnarray*} 
I_3&=& \mathbb{E}\sup_{0 \leq t \leq r}\Big|\sum_{0\leq  u\leq t, \Delta Z(u)\neq 0}\ln(1+g(S(u-b))Y_{N(u)})-\ln(1+g(v_2(u))Y_{N(u)})\Big|^{2p}\\
&=& \mathbb{E}\sup_{0 \leq t \leq r}\Big|\displaystyle\int_{\JJ }\displaystyle\int_0^t\left[ \ln(1+zg(S(u-b)))-\ln(1+zg(v_2(u)))\right] \tilde{N}(du,dz)\\
&&
\qquad +\displaystyle\int_{\JJ }\displaystyle\int_0^t\left[\ln(1+zg(S(u-b)))-\ln(1+zg(v_2(u)))\right]\nu(dz)du\Big|^{2p}\\
&=& 4^p \mathbb{E}\sup_{0 \leq t \leq r}\Big|\displaystyle\int_{\JJ }\displaystyle\int_0^t\left[\ln(1+zg(S(u-b)))-\ln(1+zg(v_2(u)))\right] \tilde{N}(du,dz)\Big|^{2p} \\
&&
\qquad +4^p \mathbb{E}\sup_{0 \leq t \leq r}\Big|\int_{\JJ}\displaystyle\int_0^t\left[ \ln(1+zg(S(u-b)))-\ln(1+zg(v_2(u)))\right] \nu(dz)du\Big|^{2p}\\
&=:&I_{31}+I_{32}\,. 
 \end{eqnarray*}
Using the Lipschitz condition on  $g$ and   (A3), we 
have  
 \begin{eqnarray*}
I_{32} 
&\le&
% \mathbb{E}\Big[\sup_{0 \leq t \leq r}|S(t)-p(t) |^{2p} \Big]\nonumber\\&& \leq 2^{(p-1)}\Big[\mathbb{E} \displaystyle\int_{0}^{T}\Big|X(u-b)-m(u-b)\Big|^{2p}du+\mathbb{E} \displaystyle\int_{0}^{T}\Big|m(u-b)-v_2(u)\Big|^{2p}du\Big]\nonumber
% \\&& +2^{2(p-1)}(\lambda  L)^{2p}\Big[\mathbb{E} \displaystyle\int_{0}^{T}\Big|X(u-b)-m(u-b)\Big|^{2p}du+\mathbb{E} \displaystyle\int_{0}^{T}\Big|m(u-b)-v_2(u)\Big|^{2p}du\Big]\nonumber\\&&+ 2^{5(p-1)}
K_p \mathbb{E}\Big(\displaystyle\int_0^r\Big|g(S(u-b))-g(v_2(u))\Big| du\Big)  ^{2p} \\
&\le& K_{p, T}  \mathbb{E} \sup_{0\le t\le r} \left|S(t-b))- S^{\pi}(t-b) \right|   ^{2p} \,. 
 \end{eqnarray*}
% 
% hence we would have \eqref{pema110} as 
% \begin{eqnarray*}&&
% \mathbb{E}\Big[\sup_{0 \leq t \leq r}|S(t)-p(t) |^{2p} \Big]\nonumber\\&& \leq 2^{(p-1)}\Big[\mathbb{E} \displaystyle\int_{0}^{T}\Big|X(u-b)-m(u-b)\Big|^{2p}du+\mathbb{E} \displaystyle\int_{0}^{T}\Big|m(u-b)-v_2(u)\Big|^{2p}du\Big]\nonumber
% \\&& +2^{2(p-1)}(\lambda L)^{2p}\Big[\mathbb{E} \displaystyle\int_{0}^{T}\Big|X(u-b)-m(u-b)\Big|^{2p}du+\mathbb{E} \displaystyle\int_{0}^{T}\Big|m(u-b)-v_2(u)\Big|^{2p}du\Big]\nonumber\\&&+ 2^{4(p-1)}\mathbb{E}\sup_{0 \leq t \leq r}\Big|\displaystyle\int_{\mathbb{R}_0}\displaystyle\int_0^t\ln(1+zg(X(u-b))-\ln(1+zg(v_2(u)))\tilde{N}(du,dz)\\&&+\displaystyle\int_{\mathbb{R}_0}\displaystyle\int_0^t\ln(1+zg(X(u-b))-\ln(1+zg(v_2(u)))\nu(dz)du\Big|^{2p}
%  \label{pema111} 
% \end{eqnarray*}
Using the Burkholder-Davis-Gundy inequality we   have
 \begin{eqnarray*}
 I_{31}&\le &
% \mathbb{E}\Big[\sup_{0 \leq t \leq r}|S(t)-p(t) |^{2p} \Big]\nonumber\\&&
% \leq 2^{(p-1)}\Big[\mathbb{E} \displaystyle\int_{0}^{T}\Big|X(u-b)-m(u-b)\Big|^{2p}du+\mathbb{E} \displaystyle\int_{0}^{T}\Big|m(u-b)-v_2(u)\Big|^{2p}du\Big]\nonumber
% \\
% && +2^{2(p-1)}(\lambda L)^{2p}\Big[\mathbb{E} \displaystyle\int_{0}^{T}\Big|X(u-b)-m(u-b)\Big|^{2p}du+\mathbb{E} \displaystyle\int_{0}^{T}\Big|m(u-b)-v_2(u)\Big|^{2p}du\Big]\nonumber\\
K_p\mathbb{E}\Big(\displaystyle\int_{\JJ }\displaystyle\int_0^r\Big|\ln(1+zg(S(u-b)))-\ln(1+zg(v_2(u)))\Big|^2\nu(dz)du\Big)^{p}\,. 
% &&+2^{p-1}.2^{4(p-1)}\mathbb{E}\Big(\displaystyle\int_{\mathbb{R}_0}\displaystyle\int_0^T\Big|\ln(1+zg(X(u-b))-\ln(1+zg(v_2(u))))\Big|\nu(dz)du\Big)^{2p}
 \end{eqnarray*}
Similar to the bound for $I_{32}$, we have
\[
I_{31}\le K_{p, T}  \mathbb{E} \sup_{0\le t\le r} \left|S(t-b))- S^{\pi}(t-b) \right|   ^{2p}\,.
\]
Combining the estimates for $I_{31}$ and $_{32}$, we see
\begin{eqnarray} 
I_3\le 
K_{p, T}  \mathbb{E} \sup_{0\le t\le r} \left|S(t-b))- S^{\pi}(t-b) \right|   ^{2p}\,.\label{e.3.13} 
\end{eqnarray}  
It is easy to verify 
\begin{eqnarray}
I_1\le 
K_{p, T}  \mathbb{E} \sup_{0\le t\le r} \left|S(t-b))- S^{\pi}(t-b) \right|   ^{2p}\,. \label{e.3.14} 
\end{eqnarray}
Inserting  the bounds obtained in \eqref{e.3.12}-\eqref{e.3.14}
into \eqref{e.3.11},  we see that
\begin{eqnarray}
I \le 
K_{p, T}  \mathbb{E} \sup_{0\le t\le r} \left|S(t-b))- S^{\pi}(t-b) \right|   ^{2p} +K_{P, T} \De^p\,. \label{e.3.14} 
\end{eqnarray}
Combining this estimate with \eqref{e.3.10a}, we see 
 \begin{eqnarray}
 &&
 \mathbb{E} \Big[\sup_{0 \leq t \leq  r}
 |S(t)-S^\pi(t)  |^p \Big] 
 \nonumber\\
 &&\qquad \le K_{p, T}  \left[ \mathbb{E} \sup_{0\le t\le r} \left|S(t-b) - S^{\pi}(t-b) \right|   ^{2p}\right]^{1/2}  +K_{P, T} \De^{p/2}   \label{e.3.17} 
\end{eqnarray}
for any $p\ge 2$ and for any $r\in [0, T]$.  Now we shall use \eqref{e.3.17} to prove the theorem   on the interval $[0, kb]$ recursively for $k=1, 2, \cdots,
[\frac{T}{b}]+1$.  Since 
$S^\pi(t)=S(t)=\phi(t)$ for $t\in [-b, 0]$.  
Taking $r=b$, we have
\begin{eqnarray} 
 \mathbb{E} \Big[\sup_{0 \leq t \leq  b}
 |S(t)-S^\pi(t)  |^p \Big]  \le K_{p,T}  \De^{p/2}   
\end{eqnarray}
for any $p\ge 2$.  Now taking $r=2b$ in \eqref{e.3.17}, we have
\begin{eqnarray}
 &&
 \mathbb{E} \Big[\sup_{0 \leq t \leq  2b }
 |S(t)-S^\pi(t)  |^p \Big] 
 \nonumber\\
 &&\qquad \le K_{p, T}  \left[ \mathbb{E} \sup_{-b\le t\le b} \left|S(t))- S^{\pi}(t) \right|   ^{2p}\right]^{1/2}  +K_{P, T} \De^{p/2}   \nonumber\\
 &&\qquad \le K_{p, T}  \left[ K_{2p, T} \De^p    \right]^{1/2}  +K_{P, T} \De^{p/2}  \le K_{p, T} \De^{p/2}\,.  
\end{eqnarray} 
Continuing this way we obtain for any positive integer $k\in \NN$, 
\begin{eqnarray}
 &&
 \mathbb{E} \Big[\sup_{0 \leq t \leq  kb }
 |S(t)-S^\pi(t)  |^p \Big]   \le  K_{k, p, T} \De^{p/2}\,.  
\end{eqnarray} 
Now since $T$ is finite, we can choose a $k$ such that
$(k-1)b<T\le kb$. This completes the proof of  the theorem. 
 \end{proof}
 
  \section{Option Pricing in Delayed Black-Scholes 
  market with jumps} 
  \label{s.4} 
 In this section we consider the problem of option pricing in a delayed Black-Scholes market which consists of two assets. One is   risk free, whose price is described  by 
\begin{eqnarray}
dB(t)=rB(t) dt\,,\quad {\rm or}\quad B(t)=e^{rt}\,, t\ge 0\,. 
\label{e.4.1}
\end{eqnarray}
Another asset is a risky one, whose price is described by the delayed equation \eqref{e.2.1} or 
 \eqref{e.3.1}, namely, 
\begin{eqnarray}
 dS(t) &=& f(S(t-b))S(t)dt +   g(S(t-b)) S(t-)d {Z}(t)\,, 
 \label{e.4.2}
 \end{eqnarray} 
where 
$
Z(t)=     \sum_{i=1}^{N_t}Y_i  
$ is a L\'evy process,     $N_t$ is  a 
 Poisson process with intensity $\la$,  and $Y_1, Y_2, \cdots,
 $ are iid random variables.  As in Section 2, we introduce the Poisson random measure $N(dt, dz)$ and its compensator
 $\tilde N(dt, dz)$. 
 The above delayed equation can be written as
 \begin{eqnarray*}
dS(t)&=& \left[ f(S(t-b)) + g(  S(t-b)) \int_{\JJ} z\nu(dz)\right]
S(t) dt\nonumber\\
&&\qquad\quad  +g(  S(t-b)) S(t-) \int_{\JJ}    z \tilde N(dt, dz) \,. 
\end{eqnarray*} 
Denote 
\begin{equation}
L=  \int_{\JJ}zf_Y(z)dz \,, 
\end{equation}
where $f_Y$ is the probability density of $Y_i$
(whose support is  $\JJ$).  Then
\[
\int_{\JJ} z\nu(dz)=\la L\,.
\]
Set 
\[
\tilde{S}(t)=\frac{S(t)}{B(t)}\,.
\]
 Then by   It\^o's formula we have 
 \begin{eqnarray}
 d\tilde{S}(t) = \tilde{S}(t-)g(S(t-b))\Big(\displaystyle\int_{ \JJ  } z\big[\theta (t)\nu (dz)dt 
 + \tilde{N} (dt,dz)\big]\Big) \,, \label{dbs11}
 \end{eqnarray} 
where $\theta(t) = \frac{f(S(t-b))+g(S(t-b))-r}{\lambda Lg(S(t-b))} $.
%\footnote{\color{blue}I have edited $\theta(t)$, A6}
We shall keep the assumptions (A1)-(A4) made in 
previous  section and we need to make an additional assumption: 
\begin{enumerate} 
%     \item[{\bf (A5)}] $f(t)>0 $  for all $t \in [0,\infty) $.
% \item[I2]  $g :\mathbb{R} \rightarrow \mathbb{R}-0$ is continuous and $|g| \leq B$ where $0<B< \infty $.
% \item[I3] $a$ and $b$ are positive constants.
% \item[I4]{\color{red} $zg(S(t-b))>-1$ for all $t\geq 0$ and $\Big|z\Big|< \frac{1}{B}$.} 
\item[{\bf (A5)}]  There is a constant $\al_1\in (1, \infty)$ 
such that 
  $\displaystyle\int_{\JJ}\nu(dz)\ge \al_1 \Big|\frac{f(s)+g(s)-r}{g(t)}\Big|$ \quad $\forall \ s, t \in [0,\infty)$ 
 \end{enumerate}
%\begin{remark}  In the model \eqref{e.4.2}, $f$ is the ``mean return of the risky asset. In the reality,  we 
%should have $f(t)>r>0$.
%Otherwise, investor will not take the risk. 
%%So, it is reasonable to assume   (A5). 
%The assumption (A5) is a technical one. 
%\end{remark}  
 To find the risk neutral probability measure we apply  Girsanov theorem for L\'evy process (see \cite[Theorem 12.21]{NOP}).  \ 
 The $\theta(t)$ is predictable for $t \in [0,T]$. From the assumptions above we also have that $0<\theta(s) \leq \frac{1}{\al_1} $.
 Thus,
  $$\displaystyle\int_{[0,T] \times \JJ } \Big(|\log(1+ \theta(s) )|+\theta^2(s) \Big)\nu (dz)ds\le K<
  \infty \,. 
  $$ 
 Now    define 
 \begin{eqnarray*}
  S^\th(t)  &:=& \exp\Big( \displaystyle\int_{ [0,t]}\{\log\big(1-\theta(s)\big) + \theta(s) \}\nu (dx)ds \\&&+ \displaystyle\int_{ [0,t]}\log\big(1-\theta(s)\big) \tilde{N} (dx,ds)\Big)\,. 
 \end{eqnarray*} 
 In order for us to obtain  an equivalent martingale measure we need to  verify  the following Novikov condition:
 \begin{eqnarray}
 \mathbb{E} \Big[ \exp\Big(\frac{1}{2}\displaystyle\int_{[0,T]\times\JJ }\{(1-\theta(s))\log(1-\theta(s)) +\theta(s)\}\nu (dz)ds \Big)\Big]<\infty  \label{dbs19}
 \end{eqnarray}
This is a consequence of our assumption (A5). In fact, we have  first
 \begin{eqnarray*}
| \theta(s)| &=& \frac{|f(S(t-b))-r|}{\lambda Lg(S(t-b))}
%<\frac{f(S(t-b))}{\lambda Lg(S(t-b))}
\le \frac{1}{\al_1}<1\,.  
 \end{eqnarray*}
% Moreover, 
%{\color{red} \begin{eqnarray*}
% -1<\theta(s) = \frac{f(S(t-b))-r}{\lambda Lg(S(t-b))}
% \end{eqnarray*}
%}
Hence we   have 
 \begin{eqnarray*}
 \displaystyle\int_{[0,T]}\{(1-\theta(s))\log(1-\theta(s)) +\theta(s)\}ds&<& \infty\,. 
 \end{eqnarray*}
But $ \nu (dz) =\lambda f_Y(z)dz$, we have 
 \begin{eqnarray*}
 \displaystyle\int_{\JJ }\nu (dz)&=&\displaystyle\int_{\JJ }\lambda f_Y(z)dz <\infty\,. 
 \end{eqnarray*}
Thus, we have  \eqref{dbs19}.
 
Now since we have  verified the Novikov condition \eqref{dbs19} we
 have then $\mathbb{E} [S^\theta(T) ] =1$. 
Define an  equivalent probability measure $\mathbb{Q}$ on $\mathcal{F}_T$   by 
\begin{eqnarray}
d\mathbb{Q}:=S^\th (T) d\mathbb{P}\,. \label{e.4.6}
\end{eqnarray} 
On the new probability space $(\Om, \cF_T, \mathbb{Q})$ (new probability $\mathbb{Q}$) the random measure  
  \begin{eqnarray}
  \tilde{N}_{\mathbb{Q}}(dz,ds)=\theta(t)\nu (dz)ds+\tilde{N} (dz,ds)\,,  \label{dbs21}
  \end{eqnarray} 
%  where $\tilde{N}_{\mathbb{Q}}$ 
  is a  compensated Poisson random measure.
   The corresponding L\'evy measure is 
  denoted by   $\nu_{\mathbb{Q}}$. 
%  \footnote{what is the relation between $\nu_{\mathbb{Q}}$ and $\nu   $? 
%  
%  ({\color{blue}Sir, As we discussed you suggested that $\nu_{\mathbb{Q}} = \nu$ 
%  
%  By the relation (\ref{dbs21})we can get $$\theta(t)\nu(dz)dt = \mathbb{E}_{\mathbb{Q}}\Big[\frac{1}{S^{\theta}(T)}N_{\mathbb{Q}}(dz,dt)\Big] - \nu_{\mathbb{Q}}(dz)dt $$) }}
With this new Poisson random measure we can write \eqref{dbs11} as 
 \begin{eqnarray}
 d\tilde{S}(t) = \tilde{S}(t-)\displaystyle\int_{ \JJ  } zg(S(t-b))\tilde{N}_{\mathbb{Q}}(dt,dz)\,.  \label{dbs20}
 \end{eqnarray} 
 The following result gives    the fair price formula for the
 European  call  option as well as  the 
 corresponding hedging portfolio.
 
  \begin{theorem}{\label{dbs16}}
Let  the  market  be given by \eqref{e.4.1} and  \eqref{e.4.2},
where the coefficients $f$ and $g$ satisfy the assumptions 
(A1)-(A5). %\deleted{The} \added{Then} 
Then the market is complete. 
Let $T$  be the maturity time of the European call
option on the stock with payoff function given by $ X = (S_T - K)^+$.
%; $ (S_T - K)^+$ is an $\mathcal{F}_t $ measurable non-negative integrable random variable. 
Then at any time $t \in [0,T]$, the fair price V(t) of the option is given by the formula
 \begin{eqnarray}
 V(t) = e^{-r(T-t)}\mathbb{E}_{\mathbb{Q}}\Big((S_T - K)^+|\mathcal{F}_t  \Big) \label{dbs13}
 \end{eqnarray}
 where $\mathbb{Q}$ is the martingale measure on $(\Omega, \mathcal{F}_T)$ 
 given by \eqref{e.4.6}. 
 %\begin{eqnarray*}
 %&Z(t) = exp{\{ \displaystyle\int_{ [0,t]}\{\log\big(1-\theta(s)\big) + \theta(s) \}\nu (dx)ds + \displaystyle\int_{ [0,t]}\log\big(1-\theta(s)\big) \tilde{N} (dx,ds)\}}
 %\end{eqnarray*} for $t \in [0,T]$. The measure $\mathbb{Q}$ is a martingale measure and the market is complete. 
 \\
 Moreover, if $\displaystyle\int_{\JJ }z^j \nu_{\mathbb{Q}}(dz) < \infty,\displaystyle\int_{\mathbb{R}_+}g(t)^jdt < \infty $ for $j=1,2,3,4$,
 % where $g$ is a continuous function; 
  there is an adapted and square integrable process $\psi(z,t) \in \mathcal{L}^2(\JJ \times [0,T])$ such  that 
 $$\mathbb{E}_{\mathbb{Q}}\Big(e^{-rT}(S_T-K)^+|\mathcal{F}_t  \Big) = \mathbb{E}_{\mathbb{Q}}\Big(e^{-rT}(S_T-K)^+\Big) + \displaystyle\int_{[0,t] \times \JJ } \psi(z,s)\tilde{N}_{\mathbb{Q}}((dz,ds)  $$
 and the hedging strategy is given by
 \begin{eqnarray}
 \pi_S(t):= \frac{\displaystyle\int_{\JJ }\psi(z,t)\tilde{N}_{\mathbb{Q}}(dz,t) }{\tilde{S}(t)g(S(t-b)) }, \hspace{5mm}\pi_B(t) := U(t) - \pi_S(t)\tilde{S}(t), \hspace{4mm} t \in [0,T]\,,  \label{dbs14}
 \end{eqnarray}
 where $U(t)=\mathbb{E}_{\mathbb{Q}}(e^{-rT}(S_T - K)^+|\mathcal{F}_t ) $. 
  \end{theorem}
 \begin{proof}
 Applying the It\^o formula to \eqref{dbs20} we get
 \begin{eqnarray}
 & \tilde{S}(T) = \exp\Big(\displaystyle\int_{[0,T] \times \JJ } \{\ln(1+zg(S(t-b)))-zg(S(t-b)\}\nu_{\mathbb{Q}}(dz)dt  \nonumber\\& \displaystyle
 +\int_{[0,T] \times \JJ } \ln(1+zg(S(t-b)))\Tilde{N}_{\mathbb{Q}}(dt,dz) \Big)
 \label{e.4.10} 
 \end{eqnarray}
% Our goal now is to find the hedging portfolio and the call option price.
% \\
Denote $X=(S_T - K)^+$ and  consider 
% the filtrations $\mathcal{F}_t $, $\mathcal{F}_t^{\tilde{N} }$, $\mathcal{F}_t^{\tilde{N}_{\mathbb{Q}}}$  which are the sigma algebra generated by $\{S_u: u \leq t \}, \{\tilde{N} (\JJ ,u): u\leq t\},\{\tilde{N}_{\mathbb{Q}}(\JJ ,u): u\leq t\}$ respectively. Then we will have $\mathcal{F}_t $=$\mathcal{F}_t^{\tilde{N} }$=$\mathcal{F}_t^{\tilde{N}_{\mathbb{Q}}}$. If $X=(S_T-K)^+$ is the contingent claim then $X$ is $\mathcal{F}_t $ measurable and then we consider the $\mathbb{Q}$ martingale 
 $$U(t):= \mathbb{E}_{\mathbb{Q}}(e^{-rT}X|\mathcal{F}_t )\,. $$
 
 In order to apply martingale representation theorem for  L\'evy  process (see e.g. \cite[Theorem 5.3.5]{d.applebaum}) we shall  first  show that $U_t \in \mathcal{L}^2$, which is implied by $\mathbb{E}_{\mathbb{Q}}[S^2_T] < \infty$. 
 %We will then find the hedging portfolio. 
  
Write $h=g(S(t-b))$.  Then we can write
  \begin{eqnarray}
 & &\tilde{S}^2_T = \exp\Big(\displaystyle\int_{[0,T] \times \JJ } \{\ln(1+zh)^2-2zh\}\nu_{\mathbb{Q}}(dz)dt \nonumber \\
  &&\qquad\qquad   \displaystyle +\int_{[0,T] \times \JJ } \ln(1+zh)^2\Tilde{N}_{\mathbb{Q}}(dt,dz) \Big)\,.  \label{dbs12}
  \end{eqnarray}
Denoting  $\mathbb{T} = [0,T] \times \JJ $ and   taking $\tilde{h} = \frac{(1+zh)^4-1}{z} $ we   have  
 \begin{eqnarray*}
 \tilde{S}^2_T&=& \exp\Big(  \frac{1}{2} \displaystyle\int_{\mathbb{T}} \{\ln(1+z\tilde{h})-z\tilde{h}\}\nu_{\mathbb{Q}}(dz)dt+ \frac{1}{2}\displaystyle\int_{ \mathbb{T} } \ln(1+z\tilde{h})\Tilde{N}_{\mathbb{Q}}(dt,dz) \Big). \\ 
 &&\qquad \qquad \exp\Big(\displaystyle\int_{ \mathbb{T} }\Big(  \frac{z\tilde{h}}{2}-zh \Big)\nu_{\mathbb{Q}}(dz)dt \Big)\,. 
 \end{eqnarray*} 
Applying the H\"older inequality we have  
 \begin{eqnarray*}
 && \mathbb{E}_{\mathbb{Q}}\big[ \tilde{S}^2_T\big]\\
 && \quad \leq \Big[  \mathbb{E}_{\mathbb{Q}}  \exp\Big(  \displaystyle\int_{ \mathbb{T} } \{\ln(1+z\tilde{h})-z\tilde{h}\}\nu_{\mathbb{Q}}(dz)dt+ \displaystyle\int_{ \mathbb{T} } \ln(1+z\tilde{h})\Tilde{N}_{\mathbb{Q}}(dt,dz) \Big) \Big]^{1/2}
 \\
 && \qquad \quad \cdot \Big[\mathbb{E}_{\mathbb{Q}}  \exp\Big( 2 \displaystyle\int_{ \mathbb{T} }\Big(  \frac{z\tilde{h}}{2}-zh \Big)\nu_{\mathbb{Q}}(dz)dt \Big) \Big]^{1/2}\\
 & & =\Big[\mathbb{E}_{\mathbb{Q}}  \exp\Big( 2 \displaystyle\int_{ \mathbb{T} }\Big(  \frac{z\tilde{h}}{2}-zh \Big)\nu_{\mathbb{Q}}(dz)dt \Big) \Big]^{1/2}\,. 
 \end{eqnarray*}
From the definition of $\tilde h$, we have
  $z\tilde{h} = (1+zh)^4-1$. Then 
 \begin{eqnarray*}
 z\tilde{h} -2zh = (1+zh)^4-1-2zh=z^4h^4+4z^3h^3+6z^2h^2+2zh\,. 
 \end{eqnarray*}
% 
% If we write
% \begin{eqnarray*}
% A(T)&=&  \exp\Big(   \displaystyle\int_{ \mathbb{T} } \{\ln(1+z\tilde{h})-z\tilde{h}\}\nu_{\mathbb{Q}}(dz)dt\Big) \\&&\cdot\exp \Big( \displaystyle\int_{ \mathbb{T} } \ln(1+z\tilde{h})\tilde{N}_{\mathbb{Q}}(dt,dz) \Big)
% \end{eqnarray*}
%  and  $B(T) = \Big[ \exp\Big( 2.\displaystyle\int_{ \mathbb{T} }\Big(  \frac{z\tilde{h}}{2}-zh \Big)\nu_{\mathbb{Q}}(dz)dt \Big) \Big]$
% 
% we will have $\mathbb{E}_{\mathbb{Q}}[A(T)]=1$ since\newline
% An application of It\^o formula (see e.g. \cite{NOP, protter}) yields 
% \begin{eqnarray*}
% A(T) &=& 1+ \int\limits_{0}^T\int\limits_{\JJ }A(s-)z\tilde{h}\tilde{N}_{\mathbb{Q}}(ds,dz)
% \\ && \implies\mathbb{E}_{\mathbb{Q}}[A(T)]=1 
% \end{eqnarray*}
% We want to show that $\mathbb{E}_{\mathbb{Q}}[B]< \infty$.
% 
% We can write $B$ as\\
% $B = 
Thus, 
\begin{eqnarray*}
 \mathbb{E}_{\mathbb{Q}}\big[ \tilde{S}^2_T\big] 
 \le  \exp\Big( \displaystyle\int_{ \mathbb{T} }\Big(  z^4h^4+4z^3h^3+6z^2h^2+2zh \Big)\nu_{\mathbb{Q}}(dz)dt \Big) 
 \end{eqnarray*}
which is finite by the assumptions of the theorem. 
 
% Since we will have  $\displaystyle\int_{\JJ }z^j \nu_{\mathbb{Q}}(dz) < \infty,\displaystyle\int_{\mathbb{R}}g^jdt < \infty $ for $j=1,2,3,4$ we have $\mathbb{E}_{\mathbb{Q}}[B]< \infty$. 
% \\

From the martingale representation 
theorem  (see e.g.   \cite[theorem 5.3.5]{d.applebaum})  there exists a  square integrable predictable mapping 
   $\psi: \mathbb{T} \times\Omega \rightarrow \mathbb{R} $  
  such  that  
 \begin{eqnarray*}
  U(t) = \mathbb{E}_{\mathbb{Q}}(e^{-rT} (S_T - K)^{+}) + \displaystyle\int_0^t\displaystyle\int_{\JJ }\psi(s,z)\tilde{N}(ds,dz). 
  \end{eqnarray*} 
Define
  \begin{eqnarray*}
  \pi_S(t)&:=& \frac{\displaystyle\int_{\JJ }\psi(z,t)\tilde{N}_{\mathbb{Q}}(dz,t) }{\tilde{S}(t)g(S(t-b)) }  \\
  &=& \frac{\displaystyle\int_{\JJ }\psi(z,t) \tilde S(t) g(S(t-b))  d
  \tilde S(t)}{\tilde{S}(t)g(S(t-b)) },
\\
\pi_B(t) &:=& U(t) - \pi_S(t)\tilde{S}(t), \hspace{4mm} t \in [0,T]\,. 
  \end{eqnarray*}
% We now show \eqref{dbs13}.
  Consider the strategy $\{ (\pi_B(t),\pi_S(t)):t \in [0,T]\} $ 
  to invest  $\pi_B(t) $ units in the riskyless asset $B(t)$ and 
   $\pi_S(t)$ units in the risky asset $S(t)$ at time $t$. Then the value of the portfolio at time $t$ is given by 
 \begin{eqnarray*}
 && V(t) :=  \pi_B(t)e^{rt}+\pi_S(t)S(t) = e^{rt}U(t)
 \end{eqnarray*} 
By the definition of the strategy   we see  that
 %{\color{red}
  \begin{eqnarray*}
 & dV(t)  =  \pi_B(t)de^{rt}+\pi_S(t)dS(t)= e^{rt}dU(t)+ U(t)de^{rt}\,. 
 \end{eqnarray*}
% }\footnote{explain more 
% 
% {\color{blue}(I have done some more explaination)}
% } 
 Hence the strategy is self-financing. Moreover, we   have $$V(T)= e^{rT}U(T) = (S_T - K)^+.$$ 
 Hence    the   claim (referring to the European call option) 
 is attainable stand therefore the market $\{S(t),B(t):t \in [0,T]\}$ is complete.
  \end{proof}
 
 The pricing formula \eqref{dbs13} is hard to evaluate  analytically 
 and we shall use a general Monte-Carlo method to find the approximate values.
 But when the time fall in the last delay period, namely, 
 when $t\in [T-b, T]$ we have the following analytic expression
 for the price. 
% We now try to find an explicit formula for the European call option.
  \begin{theorem}{\label{dbs18}}
 Assume  the conditions of  Theorem \ref{dbs16}.
%  Assume $V(t)$ to be the fair price of a European call option written on the stock $S$ following model \eqref{dbs5}, \eqref{dbs6} with maturity time $T$ and exercise price $K$.
   When $t\in   [T-b,T]$,  then price for the European Call option  is given by 
 \begin{eqnarray}
 V(t)&=&  e^{rt}\lim_{v \rightarrow \infty}\frac{1}{2\pi }\displaystyle\int_{-\infty}^{\infty}\frac{1}{ i\xi}(e^{iv\xi}-e^{iw\xi})A(t)\cdot \tilde{S}(t)\exp\left\{  \displaystyle\int_t^T\displaystyle\int_{\JJ }\Big(
  (1+zg(S(u-b)))^{( 1-i\xi)}\right.\nonumber \\
 &&- \left.
 ( 1-i\xi)\ln (1+zg(S(u-b))) -1\Big)\nu_{\mathbb{Q}}(dz)du \right\}\nonumber \\
 && -K e^{rt}\lim_{v \rightarrow \infty}\frac{1}{2\pi }\displaystyle\int_{-\infty}^{\infty}\frac{1}{ i\xi}(e^{iv\xi}-e^{iw\xi})A(t)\cdot \tilde{S}(t)\exp\left\{  \displaystyle\int_t^T\displaystyle\int_{\JJ }\Big(
  (1+zg(S(u-b)))^{ -i\xi }\right.\nonumber \\
 &&+ \left.
   i\xi \ln (1+zg(S(u-b))) -1\Big)\nu_{\mathbb{Q}}(dz)du \right\} \,,  
    \label{e.4.12}
 \end{eqnarray} 
 where  $w=\ln(K/A)-rT $   and 
\begin{eqnarray}
 A(t) = \exp\Big(\displaystyle\int_t^T\displaystyle\int_{\JJ }\{\ln{(1+zg(S(u-b)))}-zg(S(u-b))\nu_{\mathbb{Q}}(dz)du\Big)\,. \label{e.4.13} 
 \end{eqnarray}
 %, \hspace{1mm} u \in [T-l,T]$$
% For $t \in [0,T-l]$
% \begin{eqnarray*}
% & V(t)=e^{rt}\mathbb{E}_{\mathbb{Q}}[(\Tilde{S}(T)-Ke^{-rT})^+|\mathcal{F}_t ] 
% \end{eqnarray*}
  \end{theorem}
\begin{proof} 
  By \eqref{dbs13}
  %, the fair price for the  European call option is given by \begin{eqnarray*}
% V(t)&=& e^{rt}\mathbb{E}_{\mathbb{Q}}\Big((\tilde{S}(T) - Ke^{-rT})^+|\mathcal{F}_t  \Big)\nonumber
% \end{eqnarray*}
 for any time $t \in [0,T]$
% . We first note that the $\tilde{S}(t)$ is measurable with respect to $\mathcal{F}_t$ and so is $ \ln{(1+zg(S(u-b)))}$ for $u \in [T-l,T]$ and therefore the price in this case would be
we have 
 \begin{eqnarray}
 V(t)&=&e^{-r(T-t)}\mathbb{E}_{\mathbb{Q}}\Big((S(T) - K)^+ \mid \mathcal{F}_t  \Big)\nonumber\\
 & =& e^{rt}\mathbb{E}_{\mathbb{Q}}\Big((\tilde{S}(T) - Ke^{-rT})^+ \mid\mathcal{F}_t \Big)\nonumber\\ 
 & = &  e^{rt}\mathbb{E}_{\mathbb{Q}}\Big(\tilde{S}(T)\111_{\{ \tilde{S}(T) \geq Ke^{-rT} \}}\mid \mathcal{F}_t \Big) - Ke^{rt}\mathbb{Q}(\tilde{S}(T) \geq Ke^{-rT})\nonumber\\
& =:&V_1(t)-V_2(t)\,.  \label{dbs17}
 \end{eqnarray} 
 
 First,  let us compute $V_1(t)$  and  $V_2(t)$ can be computed similarly. 
The solution $\tilde S(t)$ is given by \eqref{e.4.10}, which we rewrite here:  
 \begin{eqnarray}
  \tilde{S}(T) &=&\tilde{S}(t)\exp\Big\{\displaystyle\int_t^T\displaystyle\int_{\JJ }\{\ln{(1+zg(S(u-b)))}-zg(S(u-b))\}\nu_{\mathbb{Q}}(dz)du  \nonumber \\ && \qquad \quad \displaystyle +\int_t^T\displaystyle\int_{\JJ }\ln{(1+zg(S(u-b)))}\tilde{N}_{\mathbb{Q}}(dz,du) \Big\}\,. 
 \end{eqnarray}
 When $u\in [t, T]$ and $t\in [T-b, T]$, we see that   $S(u-b)$  
 is $\cF_t$-measurable.  Hence while computing the conditional expectation of $h(\tilde S(T))$ with respect to $\cF_t$, we can consider  the integrands 
$  \ln (1+zg(S(u-b)))$  and $\ln (1+zg(S(u-b))) -zg(S(u-b))   $  as ``deterministic" functions.  Thus, the analytic expression for the conditional 
expectation is possible. But it is still complicated.
To find the exact  expression and to simplify the presentation, let us 
use the notation \eqref{e.4.13} and introduce 
%\[
%A(t) = \exp\Big(\displaystyle\int_t^T\displaystyle\int_{\JJ } 
%\left[ \ln{(1+zg(S(u-b)))}-zg(S(u-b))\right] \nu_{\mathbb{Q}}(dz)du\Big) 
%\]
% and 
 \[
 Y = \displaystyle\int_t^T\displaystyle\int_{\JJ }\ln{(1+zg(S(u-b)))}\Tilde{N}_{\mathbb{Q}}(dz,du)\,.
 \]
 With these notation we have 
 \begin{eqnarray*}
 \Tilde{S}(T)=\tilde{S}(t) A\exp{Y}\,. 
 \end{eqnarray*} 
 To calculate $\mathbb{E}_{\mathbb{Q}}\Big(e^Y\111_{\{v\geq Y \geq w \}}\Big)$ we first express $\111_{[w,v]}$   as 
 the (inverse) Fourier transform of  exponential function
because  $\EE (e^{i \xi Y})$ is computable. Since the  Fourier transform of $\111_{\{w,v \}}$ is 
 $$ \displaystyle\int_{-\infty}^{\infty}e^{ix\xi}\111_{[w,v]}dx = \frac{1}{i\xi}(e^{iv\xi}-e^{iw\xi})$$
 we can write
 \begin{eqnarray*}
 \111_{[w,v]}(x) = \frac{1}{2\pi }\displaystyle\int_{-\infty}^{\infty}\frac{1}{ i\xi}(e^{i[v-x]\xi}-e^{i[w-x]\xi})d\xi\,. 
 \end{eqnarray*}
Therefore we have
 \begin{eqnarray*}
        \mathbb{E}_{\mathbb{Q}}(e^Y\111_{\{v\geq Y \geq w \}} \mid \mathcal{F}_t )&=& \frac{1}{2\pi }\displaystyle\int_{-\infty}^{\infty}\mathbb{E}_{\mathbb{Q}}\Big(\frac{1}{ i\xi}(e^{i[v-Y]\xi+Y}-e^{i[w-Y]\xi+Y})\mid  \mathcal{F}_t  \Big)d\xi\\
       &=& \frac{1}{2\pi }\displaystyle\int_{-\infty}^{\infty}\frac{1}{ i\xi}(e^{iv\xi}-e^{iw\xi}) \mathbb{E}_{\mathbb{Q}}(e^{Y(1-i\xi)}\mid  \mathcal{F}_t  )d\xi \,.
       %\frac{1}{2\pi}\displaystyle\int_{-\infty}^{\infty}\frac{1}{i\xi}(e^{ib\xi}-e^{iu\xi})\,. 
 \end{eqnarray*}
Denote $\mathbb{T}_t = [t,T] \times \JJ $.  Then we have 
 \begin{eqnarray*}
 \mathbb{E}_{\mathbb{Q}}(e^{Y-iY\xi})
 &=& \mathbb{E}_{\mathbb{Q}}\Big(\exp{ \displaystyle\int_{\mathbb{T}_t}( 1-i\xi)\ln{(1+zg(S(u-b)))}\tilde{N}(dz,du)\mid \mathcal{F}_t  }\Big)
 \\ & = &\mathbb{E}_{\mathbb{Q}}\Big(\exp{ \displaystyle\int_{\mathbb{T}_t}( 1-i\xi)\ln{(1+zg(S(u-b)))}\tilde{N}(dz,du)}\Big)     \\
 &=&\exp\Big(\displaystyle\int_{\mathbb{T}_t}\{e^{(1-i\xi)\ln(1+zg(S(u-b)) )}\\
 &&\qquad\quad -(1-i\xi)\ln(1+zg(S(u-b)))   -1   \}\nu_{\mathbb{Q}}(dz)du\Big)\\
 & =& \exp\Big(\displaystyle\int_{\mathbb{T}_t} \{(1+zg(S(u-b)))^{(1-i\xi)}\\
 &&\qquad\quad  -\ln(1+zg(S(u-b)))^{(1-i\xi) } 
 -1 \}\nu_{\mathbb{Q}}(dz)du\Big)\,. 
 \end{eqnarray*}
Hence
 \begin{eqnarray*}
 && \mathbb{E}_{\mathbb{Q}}(e^{Y}\111_{\{v\geq Y \geq w \}}\mid \mathcal{F}_t )=\frac{1}{2\pi }\displaystyle\int_{-\infty}^{\infty}\frac{1}{ i\xi}(e^{iv\xi}-e^{iw\xi})    \exp\Big(\displaystyle\int_{\mathbb{T}_t} \{(1+zg(S(u-b)))^{(1-i\xi)}\\ 
&& \qquad \quad \qquad -  \ln(1+zg(S(u-b)))^{(1-i\xi) }-1 \}\nu_{\mathbb{Q}}(dz)du\Big)d\xi\,. 
 \end{eqnarray*}
Taking  $w=\ln(K/A)-rT $, $v \rightarrow \infty$ in the above
formula we can evaluate \eqref{dbs17} as follows. 
 \begin{eqnarray*}
 V_1(t)&=& e^{rt}\mathbb{E}_{\mathbb{Q}}\Big(\tilde{S}(T)\111_{\{ \tilde{S}(T) \geq Ke^{-rT} \}} \mid \mathcal{F}_t \Big) 
% - Ke^{rt}\mathbb{Q}(\tilde{S}(T) \geq Ke^{-rT})
 \\ &&=  e^{rt}\lim_{v \rightarrow \infty}\frac{1}{2\pi }\displaystyle\int_{-\infty}^{\infty}\frac{1}{ i\xi}(e^{iv\xi}-e^{iw\xi})A\cdot\tilde{S}(t) \cdot\exp\Big(\displaystyle\int_{\mathbb{T}_t} \{(1+zg(S(u-b)))^{(1-i\xi)}\\  &&
 -\ln(1+zg(S(u-b)))^{(1-i\xi) }-1 \}\nu_{\mathbb{Q}}(dz)du\Big)d\xi   
 %-Ke^{rt}\mathbb{Q}(\tilde{S}(T) \geq Ke^{-rT})
 \\&&
 =  e^{rt}\lim_{v \rightarrow \infty}\frac{1}{2\pi }\displaystyle\int_{-\infty}^{\infty}\frac{1}{ i\xi}(e^{iv\xi}-e^{iw\xi})A\cdot \tilde{S}(t).\exp\Big(\displaystyle\int_{\mathbb{T}_t} \{(1+zg(S(u-b)))^{(1-i\xi)}\\&& -\ln(1+zg(S(u-b)))^{(1-i\xi) }-1 \}\nu_{\mathbb{Q}}(dz)du\Big)d\xi  \,. 
%-Ke^{rt}\mathbb{Q}(\tilde{S}(T) \geq Ke^{-rT})\,. 
 \end{eqnarray*}
Exactly in the same way (and now without the factor $e^Y$), we have
\begin{eqnarray*}
 V_2(t)
& &=  K e^{rt}\lim_{v \rightarrow \infty}\frac{1}{2\pi }\displaystyle\int_{-\infty}^{\infty}\frac{1}{ i\xi}(e^{iv\xi}-e^{iw\xi})A\cdot \tilde{S}(t).\exp\Big(\displaystyle\int_{\mathbb{T}_t} \{(1+zg(S(u-b)))^{   -i\xi) }\\ 
&& -\ln(1+zg(S(u-b)))^{  -i\xi } -1 \}\nu_{\mathbb{Q}}(dz)du\Big)d\xi  \,.  
 \end{eqnarray*}
 This  gives \eqref{e.4.12}. 
 \end{proof} 

 \section{Numerical attempt }\label{s.5}
 In this section we make an  attempt to carry out some numerical computations
 of our formula (3.39)  against the American
  call options Microsoft stock traded in  Questrade
 platform. 
To apply our model in the financial market, we need to 
estimate all the parameters including the delay factor $b$ 
from the real data.  To the best  of our knowledge  the theory on the parameter
estimation is still unavailable even in the case of 
the classical model of [3]. Motivated by the work of [19],
we try our best guess of the parameters in  the model  
(3.31)-(3.32).

% Below are the tables for comparing European call option prices for two models namely Black-Scholes model and delayed model with jumps with options expiring at $T=1$ month, $T=3$ months and $T=6$ months. 

 The  real market option prices   we   consider  is for 
 the  American call option  on  Microsoft stock. The data we use is from Questrade trading/investment platform on October 5, 2020 at 12:25 PM (EDT).  We take $T $ to be one, three and six months active trading period respectively. The real prices of  the options of different strike prices  are listed in the last column  of
 the three tables below. 
 
 The readers may wonder that since 
the option pricing  formulas for both our model and the classical  
Black-Scholes model are for the European call option, why we use the market price for the  American option. The reason is that we can only find the market price for the American option. On the other hand, as stated in \cite[p.251]{kwok} 
``There is no advantage to exercise an American call prematurely
when the asset received upon early exercise does not pay dividends. The early exercise
right is rendered worthless when the underlying asset does not pay dividends, so
in this case the American call has the same value as that of its European counterpart".  See also \cite[p.61, Theorem 6.1]{karatzas}. This justifies our use of the market price for the 
American option. 
 
%  
%  
%
% and compared the European call option prices for two said models with the market  
% 
%In this section we compute  the European option price by using the formula 
%\eqref{dbs13}.  
 Using Monte-Carlo simulation we calculate the prices   of European option 
  given by \eqref{dbs13}   and the analogous Black-Scholes formula obtained from the model: $dS(t)=S(t)[\alpha  dt+\sigma  dW(t)]$.  We simulate  {2000} paths
 of the solutions to both equations  using the logarithmic Euler-Maruyama scheme
 [for Black-Scholes model   the logarithmic Euler-Maruyama scheme is 
 the same by replacing the jump process 
  by Brownian motion].  In the simulations we  take the time step
  $\Delta$ to be the trading unit minute.
So when $T=1$ month,  there are 
 \[
 n=\hbox{trading hours}  \times 60 \times  \hbox{trading days}=6.5\times 60  \times 22=8580
 \]
 minutes. So $\De=\frac{1}{8580}$.  {We do the same for $T=3$ and $ T=6$.} 
  
  In our calculation  for the delayed jump model we use  the double exponential jump process as our $Y_i$'s with parameters $p =.60, q=1-p=.40,\eta=12.8,\theta =8.40$ with the intensity  $ \lambda =.03$.  The interest rate $r=.01$ is  the risk free 
   rate.  The delay factor was taken to be one day  which is    $b=\frac{6.5\times 60}{8580}$ because there are trading $6.5$ 
    hours in a trading day.  The function $f(x)$ was taken to be a fixed constant $f(x) = .1$, $g(x) = .15*\sin(x/209.11)$ and $\phi(x) = \exp(\alpha x/n) $ with $\alpha=.11$.  We choose $\alpha=.11$ since the   initial  price we have taken is $209.11$ and the predicted average price target of Microsoft stock for next one year (around 12 months from October 5, 2020) is $230$ which is $11\%$. 
    
    For the simulation of the Black-Scholes model, based on stock prices for the year 2019 we take volatility of the Microsoft stock as $\sigma=15\%$  to calculate Black-Scholes price.
  We have taken $r=1 \%$ since in the last one year the range of 10 year treasury rate has been between .52\% to 1.92\%.   
 
 The computations are summarized in the following tables. Notice
 an interesting phenomenon   that the price we obtain by  using our formula  is comparable 
to the Black-Scholes price for shorter maturities and is more closer 
to the  real market price for longer maturity.  This may be because  of our choice  
of the parameters by guessing.  

\medskip
 \begin{tabular}{ |p{2.5cm}||p{2.5cm}|p{2.5cm}|p{2.5cm}|  }
  \hline
  \multicolumn{4}{|c|}{{Call Option price comparison for $T=1$ month for Microsoft stock}} \\
  \hline
  Strike Price& Black-Scholes option price  (European) with 1 month expiration (no delay)& Option price of jump model (European) with 1 month expiration& Market Price of American option with expiration 1 month\\
  \hline
  195    &{16.27} &   {16.08} & 18.3\\
  200& {11.41}   & {11.05} & 15.15\\
  205 & {7.65}&  {6.91} & 12\\
  210 & {4.54} &  {3.62}& 9.43\\
  215  & {2.05}& {1.48}& 7\\
  220  & {.83}  & {.61} & 5.15\\
  
  \hline
 \end{tabular}

 \begin{tabular}{ |p{2.5cm}||p{2.5cm}|p{2.5cm}|p{2.5cm}|  }
   \hline
   \multicolumn{4}{|c|}{{Call Option price comparison for $T=3$ month for Microsoft stock}} \\
   \hline
   Strike Price& Black-Scholes option price  (European) with 3 month expiration (no delay) & Option price of jump model (European) with 3 month expiration&Market Price of American option with expiration 3 months \\
   \hline
   195    &{21.37} &   {21.27} & 24.40\\
   200& {16.72}  & {16.99}& 21.35\\
   205 & {13.08}&  {14.50}& 18.55\\
   210 & {9.65} &  {11.43}& 15.95\\
   215  & {6.35}& {8.58}& 13.65\\
   220  & {4.31}  & {7.51}& 11.55\\
   
   \hline
  \end{tabular}
  
  \begin{tabular}{ |p{2.5cm}||p{2.5cm}|p{2.5cm}|p{2.5cm}|  }
    \hline
    \multicolumn{4}{|c|}{{Call Option price comparison for $T=6$ month for Microsoft stock}} \\
    \hline
    Strike Price& Black-Scholes option price  (European) with 6 month expiration (no delay)& Option price of jump model (European) with 6 month expiration& Market Price of American option with expiration 6 months\\
    \hline
    195    &{28.41} &  {29.53} & 29.00\\
    200& {23.85}  & {26.11}& 26.15\\
    205 & {19.49}&  {24.44}& 23.50\\
    210 & {16.24}&  {21.15} & 21.05\\
    215  & {12.83}& {18.39}& 18.80\\
    220  & {10.58}  & {17.97}& 16.70\\
    \hline
   \end{tabular}

 \medskip 
%We  also include  a column of market price of the American call option of Microsoft stock which helps us in making sure that the values our model is giving, \deleted{is strictly less than} {is comparable to} the values of the real world market price of the American call option.  The readers may wonder that since 
%the option pricing  formulas for both our model and the classical  
%Black-Scholes model are for the European call option, why we use the market price for the  American option. The reason is that we can only find the market price for the American option. On the other hand, as stated in \cite[p.251]{kwok}
%``There is no advantage to exercise an American call prematurely
%when the asset received upon early exercise does not pay dividends. The early exercise
%right is rendered worthless when the underlying asset does not pay dividends, so
%in this case the American call has the same value as that of its European counterpart".  This justifies our use of the market price for the 
%American option. 

\section{Conclusion} 
In this paper we introduce and study a  stochastic delay equation
 with jump and  derive  a formula for the fair   price of the European call option. We assume  that the jump is dictated by a  compensated L\'evy
  process, which includes the process like asymmetric double exponential, 
  hyper-exponential jump process. In the numerical execution we  consider the asymmetric double exponential process. Furthermore,  we   
   propose a logarithmic Euler-Maruyama scheme
   (a variant of Euler-Maruyama scheme) which preserve the positivity of the approximate solutions and show that  the convergence rate of this scheme is $0.5$
   in any $L^p$ norm,
 the   optimal rate for  the classical Euler-Maruyama scheme for the  
stochastic differential equations driven by standard Brownian motion
(see e.g. \cite{cambanis}). {From the above tables   we see    that the parameters guessed   here may  not be 
the best possible values but our formula still gives a good  fit  to the real market prices  compared to the Black-Scholes formula. We note further that potential research problem of parameter estimation is still open before we can come up with the best possible simulated results.
}
%{\color{red}We  see  in the numerical execution that price for the European call option with jump model is comparable to the Black-Scholes price and less than the market value of the American Call option price.}
%\footnote{don't know exactly what is the meaning here {\color{blue} (to be discussed)
%
%I am trying to highlight two points (may be we can ignore the 2nd point)
%
%\begin{enumerate}
%\item call option prices for black scholes model and the model we considered is comparable
%\item european call option price is less than the market price of American call option. 
%\end{enumerate}  }
%}

%To make even better approximation using assymetric double exponential jump process or hyper geometric jump process a problem which lies ahead of us is to estimate the parameters we have used and also how to determine the appropriate functions $f,g$ we have used.
 
\textbf{This research was funded by  
 an NSERC discovery fund and a startup fund of University of Alberta.}

\end{document}